\documentclass{article}
\usepackage[utf8]{inputenc}

\usepackage{amsmath,amssymb,bbm,amsthm}
\usepackage{fullpage}
\usepackage{thm-restate,color,xcolor,xspace}
\usepackage{hyperref,cleveref}
\usepackage{algorithm,algorithmic}
\usepackage{thm-restate}
\usepackage{graphicx}

\newtheorem{theorem}{Theorem}[section]
\newtheorem{lemma}[theorem]{Lemma}
\newtheorem{proposition}[theorem]{Proposition}

\newtheorem{definition}[theorem]{Definition}
\newtheorem{corollary}[theorem]{Corollary}

\newtheorem{observation}[theorem]{Observation}

\newtheorem{conjecture}[theorem]{Conjecture}
\newtheorem{oq}[theorem]{Open Question}

\def \eps {\varepsilon}

\title{Hardness of Approximation in P via Short Cycle Removal:\\ Cycle Detection, Distance Oracles, and Beyond}
\author{Amir Abboud\footnote{Weizmann Institute of Science. Email: \text{amir.abboud@weizmann.ac.il}} \and Karl Bringmann\footnote{Saarland University and Max Planck Institute for Informatics, Saarland Informatics Campus, Germany. Email: \text{bringmann@cs.uni-saarland.de}} \and Seri Khoury\footnote{UC Berkeley. Email: \text{seri\text{\_}khoury@berkeley.edu}} \and Or Zamir\footnote{Institute for Advanced Study. Email: \text{orzamir@ias.edu}}}

\begin{document}

\maketitle

\begin{abstract}
    We present a new technique for efficiently removing almost all short cycles in a graph without unintentionally removing its triangles.
    Consequently, triangle finding problems do not become easy even in almost $k$-cycle free graphs, for any constant $k\geq 4$.
    
    Triangle finding is at the base of many conditional lower bounds in P, mainly for distance computation problems, and the existence of many $4$- or $5$-cycles in a worst-case instance had been the obstacle towards resolving major open questions. 
    
    \begin{itemize}
        \item \textbf{Hardness of approximation:} Are there distance oracles with $m^{1+o(1)}$ preprocessing time and $m^{o(1)}$ query time that achieve a constant approximation?
        Existing algorithms with such desirable time bounds only achieve super-constant approximation factors, while only $3-\eps$ factors were conditionally ruled out (P{\u{a}}tra{\c{s}}cu, Roditty, and Thorup; FOCS 2012). We prove that no $O(1)$ approximations are possible, assuming the $3$-SUM or APSP conjectures. In particular, we prove that $k$-approximations require $\Omega(m^{1+1/ck})$ time, which is tight up to the constant~$c$. The lower bound holds even for the \emph{offline} version where we are given the queries in advance, and extends to other problems such as dynamic shortest paths.
        \item \textbf{The $4$-Cycle problem:} An infamous open question in fine-grained complexity is to establish any surprising consequences from a subquadratic or even linear-time algorithm for detecting a $4$-cycle in a graph. This is arguably one of the simplest problems without a near-linear time algorithm nor a conditional lower bound.
        We prove that $\Omega(m^{1.1194})$ time is needed for $k$-cycle detection for all $k\geq 4$, unless we can detect a triangle in $\sqrt{n}$-degree graphs in $O(n^{2-\delta})$ time; a breakthrough that is not known to follow even from optimal matrix multiplication algorithms.

    \end{itemize}
\end{abstract}

\thispagestyle{empty}
\clearpage
\setcounter{page}{1}

\section{Introduction}


One of the most central and challenging goals in fine-grained complexity is to prove \emph{hardness of approximation} results for the many fundamental problems that we already know are hard to compute exactly.
With the exception of few results that follow from simple gadget reductions,\footnote{Similar to saying that the NP-hardness of $3$-coloring implies a $4/3-\eps$-hardness of approximation for the chromatic number.} understanding the time vs. approximation trade-off seems to require specialized fine-grained \emph{gap amplification} techniques.
As we know from the quest for \emph{NP-hardness} of approximation that started in the early 90's, such techniques are not easy to come by, and the fine-grained restrictions on the reductions can only make matters worse.

Two notable success stories, highlighted in a recent survey by Rubinstein and Vassilevska Williams \cite{RVsurvey}, are the Distributed PCP framework \cite{ARW17} based on algebraic error-correcting codes that has lead to strong results for many pair-finding type of problems \cite{rubinstein2018hardness,KLM19,CW19,Chen+19,Chen20,abboud2019dynamic,karthik2020closest}, and a graph-products technique \cite{B+21} that has lead to impressive inapproximability results for computing the diameter of a graph \cite{DVVW19,Bonnet21a,Li21,DW21,Bonnet21b,DLV21}.
Nevertheless, we are still far away from satisfactory results for many problems (see the open questions in \cite{ARW17,RVsurvey}).
Even \emph{distance computations in graphs}, an extensively studied subject in fine-grained complexity, exhibits many huge gaps. 

As a case in point, consider the open questions below for three of the most basic problems in the area, each of them with a long list of upper bounds spanning several decades: \emph{distance oracles} \cite{thorup1999undirected, cohen2001all, DHZ00, awerbuch1998near, cohen1998fast,  thorup2005approximate, baswana2006approximate, baswana2008distance, PR14,sommer2009distance,PRT12}, \emph{dynamic shortest paths} \cite{shiloach1981line,roditty2012dynamic,bernstein2011improved,henzinger2014decremental,henzinger2015unifying,chechik2018near,demetrescu2004new, bernstein2009fully,baswana2012fully, forster2021dynamic}, and \emph{shortest cycle} (girth) \cite{itai1978finding, lingas2008efficient, roditty2012subquadratic, dahlgaard2017new, ducoffe2021faster, kadaria+22}.

\begin{oq}[Distance Oracles]
\label{oq:DO}
Can we preprocess a graph in $m^{1+o(1)}$ time and answer shortest path queries in $m^{o(1)}$ time with an $O(1)$-approximation? 
What if we are given the list of all queries in advance?
\end{oq}

No known~$O(1)$-approximation algorithm can achieve the desirable time bounds in the open question.
The above references take $m^{1+\Theta(1/k)}$ preprocessing time to answer queries with a $k$-approximation in $m^{o(1)}$ time.
Meanwhile, the best conditional lower bound by P{\u{a}}tra{\c{s}}cu, Roditty, and Thorup \cite{PRT12} only rules out a $(3-\eps)$-approximation with such time bounds under a set-intersection conjecture.\footnote{Their conjecture and hardness result apply even for preprocessing algorithms with $m^{1+o(1)}$ space (and unbounded time), but higher lower bounds are not known even when restricting the time complexity.}
Existing inapproximability results higher than the $3-\eps$ barrier are either information-theoretic incompressibility arguments~\cite{bourgain1985lipschitz,matouvsek1996distortion,thorup2005approximate} and therefore only rule out $o(m)$ space bounds, or in the cell-probe model~\cite{sommer2009distance} and therefore only apply for query times up to $\log{n}$.\footnote{The latter is due to the well-known barrier of proving higher unconditional lower bounds for any problem (see \cite{PR14,PRT12}). To prove inapproximability even with the more satisfying $m^{o(1)}$ restriction on the query time, we seem to need the conditional lower bounds approach of fine-grained complexity.}


\begin{oq}[Dynamic Shortest Paths]
\label{oq:dynamic}
Can we preprocess a graph in $O(mn)$ time, then support edge-updates in $m^{o(1)}$ amortized time, and answer shortest path queries in $m^{o(1)}$ time with an $O(1)$-approximation? 
\end{oq}

Again, no known constant factor approximation meets these desirable requirements on the update and query times. It is known how to achieve update and query time $O(m^{1/k})$ with approximation factor $O(k)$ in the partially dynamic (deletions only) case \cite{chechik2018near} and approximation factor $(\log{n})^{O(k)}$ in the fully dynamic case~\cite{forster2021dynamic}.
The only conditional lower bounds are for~$(2-\varepsilon)$-approximation algorithms and they follow directly from the lower bounds for the \emph{exact} setting~\cite{RZ04esa,AV14,henzinger2015unifying}, where it is shown that distinguishing distance $2$ from $4$ is hard.\footnote{The lower bounds hold even against much higher $O(n^{1-\eps})$ update and query times, but inapproximability results with higher multiplicative factors are not known even if we demand $m^{o(1)}$ update and query times.}

\begin{oq}[Girth]
\label{oq:girth}
Can we return an $O(1)$-approximation to the girth (i.e.\ the length of the shortest cycle) in $m^{1+o(1)}$ time? 
\end{oq}

The best known approximation with $m^{1+o(1)}$ running time is super-constant; very recently, Kadaria~\emph{et al.}~\cite{kadaria+22} obtained an $O(k)$-approximation in $O(m^{1+1/k})$ time.  A lower bound for $(4/3-\eps)$-approximation follows from assuming hardness of triangle finding, as deciding if a graph has a triangle is equivalent to distinguishing between girth $3$ vs $4$.
No better lower bound is known.

\medskip

Trying to answer the above questions negatively by a lower bound leads to a common barrier; it is the \emph{short cycle barrier} discussed below.
Overcoming this barrier is related to Open Question 5 in the distributed PCP paper \cite{ARW17} asking for gap amplification techniques from conjectures other than SETH.
This is because (the exact versions of) our distance computation problems are not SETH-hard; their hardness is via  reductions from (detecting or) \emph{listing triangles} in a graph, a problem that is hard under the 3SUM or APSP conjectures, but not under SETH.

\paragraph{The Short Cycle Barrier.}
Suppose we are given a tripartite graph $G$ with parts $A,B,C$ and want to detect a triangle $a\in A, b \in B, c \in C$.\footnote{This tripartite version is equivalent to the general case of Triangle detection by a standard reduction.}
The standard reductions to distance oracles would do the following (and more or less similarly for the other problems).
We define a graph $G'$ that is obtained from $G$ by removing the $B \times C$ edges.
Then, we query for the distance in $G'$ for any pair $\{b,c\}\in E(G) \cap B\times C$ that used to be an edge in $G$.
If the distance is small, namely $2$, we conclude that $\{b,c\}$ is in a triangle in $G$, because there must be an $a\in A$ that is connected to both $b$ and $c$; this is the yes-case.
Otherwise, if the distance is larger, namely $\geq 3$, then we conclude that $\{b,c\}$ is not in any triangle in $G$ because there is no node $a\in A$ that is connected to both $b$ and $c$; this is the no-case.
Assuming a super-linear $\Omega(m^{1+\eps})$ lower bound for finding triangles (specifically for this \emph{all-edge} version) we conclude that no distance oracle with $m^{1+o(1)}$ preprocessing and $m^{o(1)}$ query time can distinguish between distance $=2$ and $\geq 3$.

To boost this result into a strong inapproximability statement, we must amplify the gap between the distances in the yes-case vs.\ the no-case.
Since the graph $G'$ is bipartite (by the assumption on $G$) we can readily observe that the distance in the no-case will actually be $\geq 4$, not just $\geq 3$, so the above construction rules out any $(2-\eps)$-approximate answers in the aforementioned time bounds.

Unfortunately from a hardness of approximation perspective, it is rather difficult to argue that the distance in the no-case must be any larger than $4$.
This is because for any pair $\{b,c\}$ the graph $G'$ is extremely likely to contain a $4$-path that makes one \emph{zigzag}, $b \to a \to b' \to a' \to c$, i.e.\ after the first step from $b$ to $a \in N(b)\cap A$, it goes back and forth once from $a\in A$ to $b' \in B$ and back to a different node in $a' \in A$, and only then goes to $c$. (See Figure~\ref{fig:shortcycle}.)
This path does not imply that $a' \in N(b)$ nor that $a \in N(c)$ and therefore does not correspond to a triangle in $G$.
Indeed, it only corresponds to a \emph{$5$-cycle} in $G$ that contains the $\{b,c\}$ edge.
This is precisely the short cycle barrier: a short cycle allows a path to make a short detour (a zigzag) and prevents us from achieving a larger gap between the yes- and no-cases.
It is also rather clear that simply subdividing edges will not work, as it increases the distance in the yes-case as well; it seems impossible to break the factor $2$ barrier with such simple tricks.

In FOCS 2010, P{\u{a}}tra{\c{s}}cu and Roditty \cite{PR14} devised an ingenious graph-products technique (conceptually similar to \cite{B+21}) to push the lower bound to approximation factors beyond $2$. 
Thinking of their construction in the terminology of triangles, their idea is to make $G'$ have $k>3$ layers by adding $k-2$ layers between $B$ and $C$ that together represent $A$.
In the yes-case where $\{b,c\}$ is in a triangle, the distance is now $k-1$, but the main advantage is that in the no-case they manage to force any path from $b$ to $c$ to make a zigzag in \emph{each of the $k-1$ layers}, making the distance $3k-2$.
For large enough $k$ this shows that distance oracles with $(3-\eps)$-approximations cannot meet the aforementioned time bounds.
In the original paper~\cite{PR14}, they could only make this approach work for small $k$ and could only prove inapproximability for factors $2\frac{2}{3}-\eps$, but in a follow-up paper with Thorup \cite{PRT12} the full potential of this approach was realized, and they established a lower bound for any $(3-\eps)$-approximations. Alas, it is clear that $3$ is the limit of this approach.
(We remark that this is a barrier even in \emph{weighted} graphs.\footnote{The case of \emph{directed} graphs is different however. For some of these problems even deciding if the distance is finite has strong lower bounds. The reason is that the directed edges can prevent zigzags.})


\begin{figure}
\centering
\includegraphics[page=1, height=150pt, trim = 0pt 0pt 400pt 0pt]{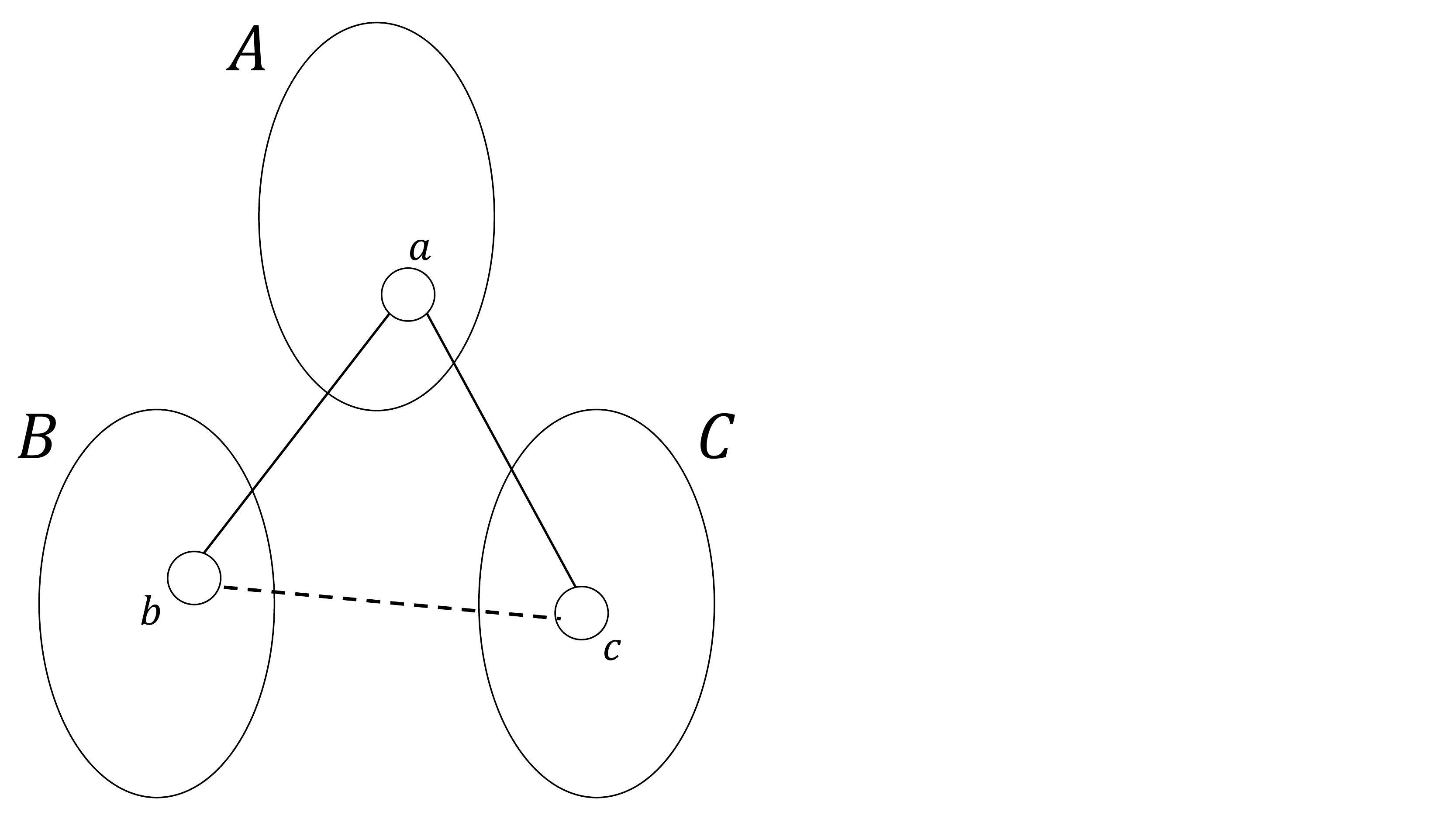}
\hspace{50pt}
\includegraphics[page=2, height=150pt, trim = 0pt 0pt 400pt 0pt]{4cyc.pdf}
\caption{If~$\{b,c\}$ belongs to a triangle in~$G$, then the distance between~$b,c$ in~$G'$ is~$2$. If~$\{b,c\}$ does not belong to a triangle in~$G$, then the distance between~$b,c$ in~$G'$ can be as small as~$4$. This corresponds to a $5$-cycle in~$G$.}
\label{fig:shortcycle}
\end{figure}

The natural and more promising approach for circumventing this barrier is to somehow ensure that there are no $\leq k$-cycles in the original graph $G$. 
Then, any effective zigzag must be \emph{long}, and even the natural two-layered construction would give us a lower bound of $\Omega(k)$.
Indeed, the distance for a pair $b,c$ would be $=2$ if the pair is in a triangle, versus $\geq k-1$ otherwise.
This would be reminiscent of the use of the girth conjecture in lower bounds for \emph{multiplicative} spanners \cite{PS89,Alt+93},
whereas the aforementioned graph-products technique is reminiscent of lower bounds for \emph{additive} spanners \cite{Woodruff06,AB17,ABP18}.\footnote{The Girth Conjecture and the techniques for additive spanners were already used, of course, for lower bounds against distance oracles as well. However, such lower bounds (and any information-theoretic arguments) cannot prove lower bounds higher than $m$; rather, they are interesting for understanding how much dense graphs can be compressed. Thus, the similarity can only be in spirit.}
All we have to do is to prove this gap amplification result for Triangle, amplifying the no-case from triangle-free to $k'$-cycle free, for all $4\leq k' \leq k$ (without unintentionally removing a triangle in the yes-case).
This boils down to the following natural question.

\begin{oq}[Main Open Question]
\label{oq:main}
Can we prove hardness for finding a triangle in a $4$-cycle free graph? What if it is $k$-cycle free for all $4 \leq k \le O(1)$?
\end{oq}

Any progress on Question~\ref{oq:main} carries over to progress on the aforementioned three open questions, by the standard reductions.
But it is far from clear why such a gap amplification should be possible.
The needle-in-a-haystack flavor (and intuitive hardness) of triangle finding \emph{stems} from the possibility of a triangle hiding amidst plenty of $4$- or $5$-cycles.
In a $4$-cycle free graph no two nodes can have more than one common neighbor; doesn't that restrict the search space by too much?\footnote{Such high-girth assumptions can indeed reduce the complexity of some problems from almost-quadratic to almost-linear. In particular, in the Orthogonal Vectors problem with dimension $d=n^{o(1)}$ (at the core of the Diameter lower bounds, and many others), if no two vectors can have two common coordinates that are non-zero, there is an $O(n d^2)$ algorithm.}

Clearly, we do not expect the triangle finding problem to remain \emph{equally hard} in $k$-cycle-free graphs as in general graphs, already because $k$-cycle-free graphs for a large even $k$ are very sparse.
Moreover, one can apply a standard reduction, e.g.\ the one to distance oracles sketched above, and then use an existing upper bound (e.g.~\cite{PRT12}) to find a triangle in $m^{1+O(1/k)}$ time.
Therefore, the main open question is whether or not the problem becomes \emph{very easy}: Can we find a triangle in a $4$-cycle free graph in linear time?
This contemplation touches upon a well-known hole in our understanding of graph problems. 
Indeed, by a simple reduction, even this latter most restricted form of Question~\ref{oq:main} is at least as hard as resolving one of the most infamous open questions in fine-grained complexity:

\begin{oq}
\label{oq:4cycle}
Can we determine if a graph contains a $4$-cycle in $m^{1+o(1)}$ time?
\end{oq}

In 1994, Yuster and Zwick \cite{YZ97} put forth the conjecture that one cannot detect a $4$-cycle in a graph in subquadratic time. 
The longstanding upper bound is $O(m^{4/3})$ via a high-degree low-degree argument 
\cite{AYZ97}. The running time can also be bounded by $O(n^2)$ because when $m \geq 200\cdot n^{1.5}$ we can simply output ``yes'': by the Bondy-Simonovitz Theorem \cite{BS74}, a graph with such density must contain a $4$-cycle.
Frustratingly, to this date, the field of fine-grained complexity has not managed to show any hardness for this problem.
``\emph{What hope do we have to
understand more complex problems if we cannot settle the complexity of this simple one, even
conditionally?}''\footnote{This is a quote from the survey by Rubinstein and Vassilevska Williams \cite{RVsurvey} where it is referring to the approximability of the graph diameter problem. We find it no less poignant when considering the $4$-cycle problem.}

\medskip

In this paper we give answers to all of the above questions, some full and some partial, based on fine-grained complexity assumptions.
It turns out that Triangle (detection or listing) requires super-linear time even when the graph has very few short cycles.

\subsection{First Result: Removing Most $k$-Cycles}

Our first main result is a fine-grained \emph{self-reduction} for Triangle from worst-case $\sqrt{n}$-degree graphs to graphs with few $k'$-cycles for all $4 \leq  k' \leq k$.
For concreteness, consider the All-Edge version where we want to report for each of the $m=O(n^{1.5})$ edges in the graph whether it is in a triangle.
This problem is known to require $n^{2-o(1)}$ time, under the $3$-SUM Conjecture \cite{Pat10,KPP16} or under the APSP Conjecture~\cite{VX20}, and this holds even for graphs of maximum degree $\sqrt{n}$. Thus it is a very plausible conjecture that $n^{2-o(1)}$ is required. (See Theorem~\ref{thm:3sumReduction} and the discussion in Section~\ref{sec:discussion}.)
A worst-case input graph to this problem might have up to $n^{k/2+1/2}$ $k$-cycles.
Given such a graph, for a sufficiently small constant $\alpha > 0$ depending on $k$, the following theorem constructs many subgraphs such that: (1) solving All-Edge-Triangle on all of these subgraphs suffices to solve the original problem, (2) the total number of edges in all these subgraphs is subquadratic $n^{2-\Omega(1)}$, and (3) the total number of $k$-cycles in all these subgraphs is subquadratic $n^{2-\Omega(1)}$. 
The latter implies that a linear-time algorithm for All-Edge-Triangle in graphs with few short cycles implies a subquadratic algorithm for the starting problem and refutes the popular conjectures.

\begin{restatable}[Removing Most $k$-Cycles]{theorem}{ThmRemovingCycles}
\label{thm:main1}

For any choice of constants~$k\geq 4, \alpha\in (0,\frac{1}{2})$, and $\eps\in (0,\frac{3-\omega}{4})$ the following holds.
Given a graph~$G$ with~$n$ vertices and maximum degree at most $\sqrt{n}$, there is a randomized algorithm, running in time $O(n^{2-\eps})$, that returns a subset of the edges $E' \subseteq E(G)$ and a collection of $s=n^{3/2-3\alpha}$ subgraphs $G_1,\ldots,G_s \subseteq G$ such that:
\begin{itemize}
    \item Every edge~$e\in E'$ participates in a triangle of~$G$.
    \item If an edge $e \in E(G)$ participates in a triangle of~$G$, then it is either in $E'$ or it participates in a triangle in at least one subgraph~$G_i$.
    \item With high probability, each $G_i$ has $O(n^{1/2+\alpha})$ vertices and maximum degree $O(n^{\alpha})$. 
    \item For every~$i$, the expected total number of $k'$-cycles of sizes $4\leq k'\leq k$ in~$G_i$ is at most~$O(n^{\frac{\omega-1}{4}+k\alpha+\eps})$.
\end{itemize}

\end{restatable}

This result achieves a weaker statement than that asked by Question~\ref{oq:main} because it does not remove all short cycles.
Still, it is sufficient for fully resolving Questions~\ref{oq:DO} and~\ref{oq:dynamic} above.
Intuitively, by applying the standard reductions (as described above), each of the few remaining short cycles might result in a false positive: a pair $\{b,c\}$ that has short distance even though it is not in a triangle.
But since the number of such cycles is small (and the degrees in the $G_i$ graphs are small), they can all be filtered in a post-processing stage in subquadratic time.

Before giving the inapproximabilty results, let us briefly explain why the matrix multiplication exponent $2\leq \omega < 2.37286$ \cite{AlmanW21} appears in our statements.
Perhaps counter-intuitively, our lower bounds get \emph{higher} the closer $\omega$ gets to $2$.
Roughly speaking, this is because our results follow from reductions that employ several procedures, including fast matrix multiplication, to extract these subgraphs with few short cycles from a given graph.
In any case, our results are new and meaningful for any $2 \leq  \omega < 2.37286$ (or even any $2 \le \omega < 3$); the only difference is in the constants.

\paragraph{Applications}
Our first corollary improves the $(3-\eps)$ hardness of P{\u{a}}tra{\c{s}}cu, Roditty, and Thorup \cite{PRT12} all the way up to $\omega(1)$, showing that $O(k)$-approximation with $O(m^{1+1/k})$ preprocessing is indeed the right tradeoff for distance oracles with $O(m^{1/k})$ query time.
Our lower bound is comparable to that of Sommer, Verbin, and Yu~\cite{sommer2009distance} in the cell-probe model, except that we allow much higher query time: $m^{\Omega(1)}$ vs. their $O(1)$. 
Moreover, our lower bound applies to the easier \emph{offline} version of the problem where all the queries are given in advance; previous lower bounds \cite{PR14,PRT12,sommer2009distance} do not apply to this restricted setting.\footnote{In the stronger models that these papers consider, where we measure space/probes rather than time, this offline problem becomes trivially easy.}
If $\omega=2$ our lower bound becomes $m^{1+\frac{1}{4k-2}-o(1)}$ time for a $(k-\delta)$-approximation; with the current $\omega$ it is $\Omega(m^{1+\frac{1}{6.3776k-4.3777}})$.

\begin{restatable}[Hardness of Approximation for Offline Distance Oracles]{corollary}{CorDistOracles}\label{cor:distOracles}
Let $k\geq 4$ be an integer, and let $\eps,\delta>0$. Define $c = \frac{4}{3-\omega}$ and $d = \frac{2\omega-2}{3-\omega}$.
Assuming either the $3$-SUM Conjecture or the APSP Conjecture, no algorithm can return a $(k-\delta)$-approximation to the distance between $m$ pairs of nodes in a simple graph with $m$ edges in time $O(m^{1+\frac{1}{ck-d}-\eps})$. 
Consequently, there is no $(k-\delta)$-approximate distance oracle with $O(m^{1+\frac{1}{ck-d}-\eps})$ preprocessing and $O(m^{\frac{1}{ck-d}-\eps})$ query time.
\end{restatable}

The \emph{Offline Distance Oracle} problem in the above corollary is at the core of the dynamic shortest paths problem as well. 
By a straightforward reduction it implies that Chechik's decremental APSP $k$-approximation algorithm in total time $O(m^{1+1/\alpha k})$ is tight up to the constant $\alpha$.
For fully dynamic algorithms, we can strengthen the result further by ruling out algorithms that start with a cubic-time preprocessing phase (which is natural as it gives the algorithm enough time to precompute all the distances). However, in this fully dynamic case, the best known upper bound only achieves a $(\log{n})^{O(k)}$-approximation in $O(m^{1+1/k})$ time.

\begin{restatable}[Hardness of Approximation for Dynamic APSP]{corollary}{CorDynamic}\label{cor:dynamic}
Let $k\geq 4$ be an integer, and let $\eps,\delta>0$,  $c = \frac{4}{3-\omega}$ and $d = \frac{2\omega-2}{3-\omega}$.
Assuming either the $3$-SUM Conjecture or the APSP Conjecture:
\begin{itemize}
 \item No algorithm can maintain a simple graph through a sequence of edge-deletion updates in a total of $O(m^{1+\frac{1}{ck-d}-\eps})$ time, while answering distance queries between a given pair of nodes with a $(k-\delta)$-approximation in $O(m^{\frac{1}{ck-d}-\eps})$ time.
    \item No algorithm can preprocess a simple graph in $O(n^3)$ time and then support (fully dynamic) updates and  queries in $O(m^{\frac{1}{ck-d}-\eps})$  time, where an answer to a query is a $(k-\delta)$-approximation to the distance between a given pair of nodes.
\end{itemize}
\end{restatable}


We next go back to the $4$-Cycle problem.
A direct corollary of our theorem is that the All-Edge version has a super-linear lower bound, finally extending the $m^{3/2}$ lower bound for triangle enumeration from P{\u{a}}tra{\c{s}}cu's seminal paper \cite{Pat10} to a hardness result for $4$-cycle enumeration.\footnote{P{\u{a}}tra{\c{s}}cu's lower bound is presented as a lower bound for the \emph{listing} problem, rather than enumeration, where we are required to list $m$ triangles (and the lower bound is $m^{4/3-o(1)}$). Our lower bound also extends to this version, but the exponent is smaller.}
If $\omega=2$ the lower bound is $m^{5/4-o(1)}$, and with the current $\omega$ it is $\Omega(m^{1.1927})$.

\begin{restatable}[Hardness for $k$-Cycle Enumeration]{corollary}{CorEnum}\label{cor:Enum}
Let $k\geq 3$ be an integer, and let $\eps>0$. Assuming either the $3$-SUM Conjecture or the APSP Conjecture, no algorithm can process an $m$-edge graph in $O(m^{1+\frac{3-\omega}{2(4-\omega)}-\eps})$ time and then enumerate $k$-cycles with $m^{o(1)}$ delay.
\end{restatable}

Unlike our two previous corollaries, the lower bound in Corollary~\ref{cor:Enum} does not get weaker with $k$.
This is because for all $k>4$ we can simply apply Theorem~\ref{thm:main1} with $k=4$ and then use known simple gadget reductions that show that $k$-cycle (detection or enumeration) for any $k$ is at least as hard as either the $k=3$ or $k=4$ case (see \cite{DKS17}).\footnote{The reduction simply subdivides some of the edges. Note that this trick is not useful in the hardness of approximation context because subdividing edges increases the distances even in the yes-case.} 
Either way, we separate $k$-cycle from the class $\textsc{DelayC}_{lin}$ of problems solvable with linear time preprocessing and constant delay~\cite{DG07}. 
This class has received significant attention in recent years from the enumeration algorithms community (see e.g. \cite{segoufin2015constant,florenzano2018constant,CK19,CK20}).
Our result is somewhat surprising because enumerating \emph{all} cycles (without restricting $k$) is \emph{in} the class $\textsc{DelayC}_{lin}$ \cite{birmele2013optimal}.

\medskip
Theorem~\ref{thm:main1} does not fully resolve the Main Open Question~\ref{oq:main}, because it does leave $n^{\Omega(1)}$ short cycles in the graph.
This is not an issue for all of the applications above because the application problem returns multiple answers (that can be filtered afterwards).
Unfortunately, this cannot be done for problems with a single output, such as our most basic $4$-Cycle detection problem.
Nonetheless, with a bit more work we can actually get rid of all $k$-cycles in the $k=4$ case, as we discuss next.

%
%

\subsection{Second Result: Removing \emph{all} $4$-Cycles}

Our most technical result is a strengthening of Theorem~\ref{thm:main1}, giving a reduction to graphs that are \emph{completely} $4$-cycle-free.
The following theorem is analogous to Theorem~\ref{thm:main1} and should be thought of in the same way; as a self-reduction from Triangle. The main two differences are that it only works for $k=4$, but it achieves the much stronger property of $4$-cycle freeness in the subgraphs it produces.

\begin{restatable}{theorem}{ThmAllFourCycles}
\label{thm:main2}
For any choice of constant~$\alpha\in (0,\frac{3-\omega}{8})$ and~$\varepsilon\in (0, \frac{3-\omega}{4}-2\alpha)$ the following holds.
Given a graph~$G$ with~$n$ vertices and maximum degree at most $\sqrt{n}$, there is a randomized algorithm, running in time $O(n^{2-\eps})$, that returns a subset of the edges $E' \subseteq E(G)$ and a collection of $s=n^{3/2-3\alpha}$ subgraphs $G_1,\ldots,G_s \subseteq G$ such that:
\begin{itemize}
    \item Every edge~$e\in E'$ participates in a triangle of~$G$.
    \item If an edge $e \in E(G)$ participates in a triangle of~$G$, then it is either in $E'$ or it participates in a triangle in at least one subgraph~$G_i$.
    \item With high probability, each $G_i$ has $O(n^{1/2+\alpha})$ vertices and maximum degree $O(n^{\alpha})$.
    \item With probability larger than~$0.99$, no subgraph~$G_i$ contains a~$4$-cycle.
\end{itemize}
\end{restatable}

This result fully resolves the main Open Question~\ref{oq:main} in the $k=4$ case.
Consequently, we improve the lower bound for girth approximation in $m^{1+o(1)}$ time from $4/3-\delta$ to $5/3-\delta$; thus making the first non-trivial step towards Open Question~\ref{oq:girth}.
And most importantly, we establish the first conditional lower bound for $4$-Cycle detection, resolving Open Question~\ref{oq:4cycle}.
If $\omega=2$ the lower bound is $m^{7/6-o(1)}$, and with the current $\omega$ it is $\Omega(m^{1.1194})$. Note, however, that Corollary~\ref{cor:fourFree} uses a less standard conjecture compared to our previous results.

\begin{restatable}[Hardness for Triangle in $4$-Cycle-Free Graphs]{corollary}{CorFourFree}\label{cor:fourFree}
Assuming that triangle detection in graphs with maximum degree at most $ \sqrt{n}$ requires $n^{2-o(1)}$ time, no algorithm can solve any of the following problems in $O(m^{1+\frac{3-\omega}{2(5-\omega)}-\eps})$ time, for any $\eps>0$:
\begin{itemize}
    \item Decide if an $m$-edge graph has a $4$-cycle.
    \item Decide if an $m$-edge $4$-cycle-free graph has a triangle.
    \item Compute a $(5/3-\delta)$-approximation to the girth of an $m$-edge graph, for any $\delta>0$.
\end{itemize}
\end{restatable}

As mentioned already, folklore gadget reductions show that either Triangle or $4$-Cycle can be reduced to a single instance of $k$-Cycle detection on the same number of edges, for any $k\geq 3$ (we add a proof of this statement in Appendix~\ref{thm:ckto34}, similar reductions appear, for example, in \cite{DKS17}).
Thus, we establish a super-linear $\Omega(m^{1.1194})$ lower bound for $k$-Cycle detection for \emph{all} constant $k\geq 3$.



Breaking the hardness assumption at the base of our conditional lower bound would be a major breakthrough.
The longstanding upper bound for triangle detection is $O(\min\{m^{2\omega/(\omega+1)},n^{\omega}\})$ \cite{AYZ97}.
Even if an optimal matrix multiplication algorithm exists ($\omega=2$) no algorithm breaks the quadratic barrier when $m=n^{1.5}$.
This continues to be the case in the natural setting where the maximum degree (rather than the average degree) is $\sqrt{n}$.
Note that we cannot base these lower bounds on 3SUM or APSP, as we did in the results above for problems with multiple outputs, until we know how to base the hardness of Triangle detection itself on these assumptions.
See Section~\ref{sec:discussion} for further discussion.

\section{Technical Overview}


The goal of this section is to give an overview of the main new ideas that go into our short cycle removal technique.
As discussed in the introduction, the technical barriers are most prominently apparent in the challenge of proving a hardness result for $4$-Cycle (detection).
For this reason, we choose to focus this section on this result, giving a tour of the reduction from Triangle (detection) to $4$-Cycle.
All of our conceptually new ideas go into this result and can be appreciated more clearly in this simple context. 
Afterwards, in Section~\ref{sec:k>4} we point out how these ideas lead to Theorems~\ref{thm:main1} and~\ref{thm:main2}.
The additional required ideas are either standard tricks (e.g. for the applications) or technical but unsurprising generalizations (e.g. for the $k>4$ case); we briefly mention them in Section~\ref{sec:k>4}.
While this presentation of results goes in the reverse order to that of the introduction, it has the advantage of presenting all important ideas while the reader needs to only think about the following deceptively simple goal: \emph{solve Triangle in $\sqrt{n}$-degree graphs in subquadratic time, given a linear-time algorithm for $4$-Cycle}.

\paragraph{Some notation:} Throughout the paper, we assume that the input graph $G=(V,E)$ for the triangle finding and all-edge triangle problems is tripartite with sides $A,B,$ and $C$.\footnote{If not, let $A=B=C=V$ and copy each edge in $E$ three times.} We denote the set of neighbors of a node $u$ by $N(u)$. We use $[k]$ to denote the set of integers $\{1,2,\cdots,k\}$, and we use the notation ``$\{4,..,k\}$-cycles'' to denote the set of cycles of length at least $4$ and at most $k$. We say that an event happens with high probability if it happens with probability at least $1-1/n^c$, for an arbitrarily large constant $c$. We denote by $\omega$ the matrix multiplication exponent.
Throughout this paper, by ``subquadratic'' we mean any bound of the form $O(n^{2-\eps})$, for any constant $\eps > 0$. Moreover, we always treat $k$ as a constant.

\subsection{Triangle to $4$-Cycle}
 To reduce Triangle to 4-Cycle, we want to convert a hard instance for Triangle in a way such that any triangle becomes a 4-cycle. Since a hard instance for Triangle is a tripartite graph with sides $A,B$, and~$C$, perhaps the first idea that one may try to use is to subdivide the edges between $B$ and $C$, by adding a dummy node $bc$ on each $\{b,c\}$ edge. Indeed, if the original graph doesn't contain a 4-cycle, then a 4-cycle in the new graph must use a dummy node, and  therefore, the existence of a 4-cycle in the new graph implies the existence of a triangle in the original graph. 

However, the original graph may have up to $n^{2.5}$ 4-cycles.\footnote{This is because each edge may be in up to $n$ $4$-cycles.} In particular, even after adding the dummy nodes between $B$ and $C$, there could still be up to $n^{2.5}$ 4-cycles between $A$ and $B$ (4-cycles that use two nodes from $A$ and two nodes from $B$), up to $n^{2.5}$ 4-cycles between $A$ and $C$, and up to $n^{2.5}$ 4-cycles that use two nodes from $A$ and one node from each of $B$ and $C$. None of these 4-cycles uses a dummy node, and therefore, the existence of a 4-cycle in the new graph doesn't necessarily imply the existence of a triangle in the original graph. We call these \emph{false $4$-cycles}.
Notably, this issue does not arise when reducing to $k$-cycle detection for odd $k$ (see \cite{DKS17}) and it can be side-stepped easily in harder contexts such as the directed\footnote{Where the directions can prevent the existence of false cycles.} or counting\footnote{By counting the number of $4$-cycles in the induced graph on subsets of the three parts we can find out the exact number of false $4$-cycles and then subtract it from the total number.} versions or in other models \cite{ACKL20,BDG07}\footnote{E.g. in databases a cycle can be forced \emph{by definition} to use one node from each of the three parts and one dummy node.}.
Alas, such simple tricks do not work in the most basic case. As discussed in the introduction, this is not a mere technicality but \emph{the} obstacle for gap amplification results.

To overcome this, one may try to remove all the $4$-cycles from the graph before applying the reduction, perhaps by finding a set of edges that intersects all of them, checking whether there is a triangle that uses one of these edges, and then removing these edges from the graph. Indeed, this would leave the graph without any $4$-cycles. However, even if we can efficiently check for each of these edges whether it is in a triangle, finding $n^{2.5}$ 4-cycles is impossible in subquadratic time. 


\paragraph{Hitting Cycles Faster Than Triangles: Random Slicing}
At a high-level, our first main idea is simple: since a $4$-cycle uses one more node than a triangle, a random subsampling reduces the number of $4$-cycles \emph{at a higher rate} than it reduces the number of triangles. Indeed, suppose that we subsample nodes, leaving each node in the graph with probability $p$. A triangle survives with probability $p^3$ while a $4$-cycle only survives with a smaller probability $p^4$.
To implement this idea we use the following \emph{random slicing} approach.

Roughly speaking, instead of reducing Triangle to 4-Cycle in the original tripartite graph, we break the graph into triangle-disjoint tripartite subgraphs,\footnote{By triangle disjoint we mean that each triangle appears in exactly one of these tripartite subgraphs.} which we refer to as \emph{slices}, and reduce Triangle to 4-Cycle in each of these slices. This way, we would only need to remove $4$-cycles that are fully contained in the slices, and not $4$-cycles \emph{across} the slices. In more detail, we partition each of $A,B$ and $C$ randomly into $n^{1/2-\alpha}$ sets, each of size $n^{1/2+\alpha}$ (where each node joins one of the sets uniformly at random). In order to solve Triangle in the original graph, it suffices to solve Triangle in each $(A_j,B_k,C_{\ell})$ slice, for $j,k,\ell\in [n^{1/2-\alpha}]$. For each such slice, we want to reduce Triangle to 4-Cycle by first removing all the $4$-cycles in the slice, and then subdividing the edges between $B_k$ and $C_{\ell}$. This time, the expected total number of $4$-cycles in all the slices is smaller than $n^{2.5}$, for $\alpha<1/2$. This is because each slice is expected to have $n^{1/2+4\alpha}$ $4$-cycles, and in total, over all slices, we have $n^{3/2-3\alpha+1/2+4\alpha}=n^{2+\alpha}$ $4$-cycles in expectation. 

Unfortunately, $n^{2+\alpha}$ $4$-cycles is still too much.
If the number of $4$-cycles is super-quadratic, then finding and removing all of them in subquadratic time is hopeless.
In other words, while we \emph{can} hit $4$-cycles at a higher rate than we hit triangles with random slicing, this higher rate is not fast enough to get the number of $4$-cycles down from the worst-case $n^{2.5}$ to subquadratic without essentially deleting all edges.  

Nevertheless, observe that if we started with fewer than $n^{2.5}$ 4-cycles in the original graph, then the expected number of $4$-cycles over all the slices following the random slicing \emph{would be} subquadratic. 
Can we prove that graphs with too many $4$-cycles are not actually hard for Triangle?\footnote{It is natural but perhaps a bit surprising that our goal switched from proving the hardness of Triangle in $4$-cycle free graphs to showing the easiness of Triangle in graphs with a maximal number of $4$-cycles.} 

\paragraph{Structure vs. Randomness: Dense-Piece Removal} Our next main observation is that random $\sqrt{n}$-regular graphs only have $O(n^2)$ $4$-cycles.
So if a graph has the worst-case $n^{2.5}$ number of $4$-cycles, then it must have a lot of structure that one could potentially exploit for solving Triangle faster.
A priori, this may not sound like a promising approach because we know that Triangle is very easy in random graphs,\footnote{In a random $\sqrt{n}$-regular graph, any edge is in a triangle with constant probability, so we can find a triangle in $O(\sqrt{n})$ expected time.} and its hardness arises from structure that existing algorithms cannot exploit.
Fortunately, we identify a connection between the existence of many $4$-cycles and the existence of dense subgraphs that we do know how to exploit algorithmically.
This is based on the fact that the savings from using fast matrix multiplication for Triangle are greater in denser graphs.
This is a novel use of the structure vs. randomness paradigm~\cite{Tao07} in the context of Triangle.\footnote{The other two examples that come to mind are the mildly subcubic combinatorial algorithm of Bansal and Williams \cite{BW12} using the Freize-Kannan regularity lemma, and the distributed algorithms (see \cite{CPSZ21}) that exploit an expander decomposition of the graph.}

In a bit more detail, we show that in subquadratic time we can find a set of edges $\tilde E$, with an answer to each edge in the graph to whether it participates in a triangle containing an edge in $\tilde E$, such that the graph induced by $E\setminus \tilde E$ has fewer than $n^{2+\gamma}$ $4$-cycles, for some $0\leq \gamma < 1/2$ to be chosen later. 
This is based on the connection between $4$-cycles and dense subgraphs. 
 One can show that an $n$-node graph with maximum degree $\sqrt{n}$ and at least $n^{2+\gamma}$ $4$-cycles must contain at least $n^{1+\gamma}$ dense subgraphs, each with $2\sqrt{n}$ nodes and roughly $n^{1/2+\gamma}$ edges.\footnote{This is because there are $n^{3/2}$ edges, and each edge participates in at most $n$ $4$-cycles.} We refer to such  subgraphs as \emph{dense pieces}. In particular, each of these dense pieces lies between two neighborhoods $N(u)$ and $N(v)$ where $\{u,v\}$ is an edge. We can use this property to find a dense piece efficiently: we sample $\approx n^{1/2-\gamma}$ edges $\{u,v\}$, and for each of them we sample $\approx n^{1/2-\gamma}$ pairs of nodes between $N(u)$ and $N(v)$ to estimate the number of edges in $N(u)\times N(v)$. With high probability, for one of the sampled  edges $\{u,v\}$, there are at least $\approx n^{1/2+\gamma}$ edges between $N(u)$ and $N(v)$, and it will be detected when we estimate the number of edges between $N(u)$ and $N(v)$. After finding the dense piece $N(u)\cup N(v)$, we use a matrix multiplication approach, together with a high-degree low-degree analysis, to efficiently check whether there is a triangle that uses an edge from the dense piece. Hence, by removing all the dense pieces gradually, where in each step we find a new dense piece, check whether there is a triangle that uses an edge from the dense piece, and then remove it, we obtain a graph with fewer than $n^{2+\gamma}$ $4$-cycles. Since in each step we remove a dense piece of $\Omega(n^{1/2+\gamma})$ edges, and the number of edges is $O(n^{3/2})$, the number of steps is bounded by $O(n^{1-\gamma})$. For an appropriate choice of $\gamma=\frac{\omega-1}{4}+\eps'<0.345$, for an arbitrarily small constant $\eps'>0$, we show that the total running time for removing all the dense pieces is subquadratic.

Hence, we first remove all dense pieces in subquadratic time, which leaves only $n^{2+\gamma}$ 4-cycles to begin with. Then, we apply the random slicing, and reduce Triangle to 4-Cycle in each of the $(A_j,B_k,C_{\ell})$ slices where we know that the total number of false $4$-cycles is subquadratic. 
This brings us to the final task of removing all of the remaining $4$-cycles in subquadratic time.
Let us remark at this point that for Theorem~\ref{thm:main1} and its applications where we can tolerate the existence of few cycles, this extra step is not necessary.

\paragraph{Output-Sensitive False Cycle Removal} Following the dense-piece removal and the random slicing, the total number of $4$-cycles in all the slices is $n^{3/2+\gamma+\alpha}$, which for some choice of $\alpha<0.155$, is subquadratic. 
It may still seem challenging to list all of them efficiently, even when their number is subquadratic: we do not even know how to find \emph{one} $4$-cycle in a general $\sqrt{n}$-degree graph in subquadratic time.

By exploiting a special property of $4$-cycles in tripartite graphs, together with the small degree property of each slice, we design an algorithm that lists all false cycles in time that is linear in their number. The crux of our idea is that all the $4$-cycles in a slice can be found by looking \emph{only at one part of the slice at a time}. For instance, to list all the $4$-cycles that use two nodes from $A_j$, it suffices to list all the $\{a,u,a'\}$ two-paths in the slice, where $\{a,a'\}\subseteq A_j$. For this, we can list all the two-paths between $A_j$ and $B_k$, and all the two-paths between $A_j$ and $C_{\ell}$, and then find all the $4$-cycles that use a pair $\{a,a'\}\subseteq A_j$ (by looking at all the two-paths that the pair $\{a,a'\}$ participates in). At first sight, since we have $n^{3/2-3\alpha}$ slices, each containing $n^{1/2+3\alpha}$ pairs $\{a,a'\}$ that can be connected by a two-path, this approach may seem to take quadratic time inevitably. Yet, with an additional trick, and a more sophisticated \emph{global} analysis that takes into account all the slices at once, we are able to charge the running time to the total number of $4$-cycles in all the slices, which is subquadratic.   

As a final clean-up step before subdviding the edges between $B_k$ and $C_\ell$ and making a call to the $4$-Cycle algorithm, we delete an edge from each of the $n^{3/2+\gamma+\alpha}$ $4$-cycles that were found.
But first, we must check whether any of them is in a triangle, and we cannot afford to spend the trivial $\sqrt{n}$ time for each.
Fortunately, we only need to look for a triangle \emph{in the slice} where nodes only have degree $n^{\alpha}$, making the total time $n^{3/2+\gamma+2\alpha}$ still subquadratic for $\alpha<0.0775$.

\medskip
Thus, in subquadratic time we can make sure that all calls that we make to the $4$-Cycle algorithm are made on graphs without false $4$-cycles.
To conclude, we point out that the total sizes of all the $4$-Cycle instances that our reduction produces is also subquadratic, so if we could solve $4$-Cycle in linear time (or even $m^{1.01}$) we would get a subquadratic algorithm for Triangle.
Indeed, the number of slices is $n^{3/2-3\alpha}$ and each has $n^{1/2+2\alpha}$ edges.

\subsection{The Theorems and Corollaries}
\label{sec:k>4}

First, let us clarify the connection between the above reduction and our Theorems for the $k=4$ case.
The slices are precisely the $G_i$ subgraphs with few cycles that our theorems produce, and the set of edges $E'$ are those edges that we identify (and remove) in the dense-piece removal process as participating in a triangle.
For Theorem~\ref{thm:main1} the final process of removing all false $4$-cycles is not necessary; the number of remaining $4$-cycles in the slices is small enough.
For Theorem~\ref{thm:main2} we do list all remaining $4$-cycles and remove an edge from each, while placing it in $E'$ if it participates in a triangle. The slices that result after this clean-up are the $4$-cycle free $G_i$ subgraphs that we return.

\paragraph{The $k>4$ Case} To remove most $k'$-cycles for all $4\leq k'\leq k$ we follow a similar route.
Even though the number of $k'$-cycles in a worst-case (or random) graph becomes larger as $k'$ grows, the random slicing method also becomes more effective at reducing their number.
The intuition is that a $k'>4$ cycle uses even more nodes than a triangle does, and so it is even less likely to survive a random subsampling; i.e. longer cycles can be hit at an even higher rate.
This leads to the same situation where we would be done if the number of $k'$-cycles was lower than the worst case, e.g. if the graph was random.
With more careful combinatorial arguments we manage to obtain a similar structure vs. randomness result: if a graph has too many $k'$-cycles then it must have a dense subgraph, and moreover such dense subgraphs can be found efficiently.
Once we have that, reducing the number of $k'$-cycles by removing the dense pieces proceeds in exactly the same way as in the $k=4$ case. 

Roughly speaking, we show that if a graph with maximum degree $\sqrt{n}$ has at least  $n^{k'/2+\gamma}$ $k'$-cycles, then it must contain many $O(\sqrt{n})$-node dense subgraphs (pieces), each with at least $\Omega(n^{1/2+\gamma})$ edges. In particular, each of these dense pieces lies between two neighborhoods of a pair of nodes $\{u,w\}$ that are connected by a simple $(k'-2)$-path (a path of $k'-2$ nodes, including $u$ and $w$). In more detail, for a simple path of $k'-2$ nodes, $\{u,u_1,\cdots, u_{k'-4},w\}$, we say that it is $\gamma$-dense if the number of $k'$-cycles that use it is $\Omega(n^{1/2+\gamma})$. Thus, a $\gamma$-dense path implies that the number of edges between $N(u)$ and $N(w)$ is $\Omega(n^{1/2+\gamma})$. We show that if a graph with maximum degree $\sqrt{n}$ has at least $n^{k'/2+\gamma}$ $k'$-cycles, then it must contain at least $\frac{1}{2}n^{k'/2-1+\gamma}$ simple $(k'-2)$-paths that are $\gamma$-dense. Since there are at most $n^{k'/2-1/2}$ $(k'-2)$-paths in the graph, we can use this property to find a dense piece efficiently, by sampling paths and estimating the density between the neighborhoods of the extreme nodes. Similar path counting arguments were also employed in the cell-probe lower bound of Sommer, Verbin, and Yu~\cite{sommer2009distance}, but the overall argument and set-up is completely different. 

This suffices for proving Theorem~\ref{thm:main1}. Substantially new ideas are required if one wishes to remove \emph{all} $k'$-cycles, extending Theorem~\ref{thm:main2} to $k>4$, because our output-sensitive enumeration strategy no longer applies.

\paragraph{The Applications} The corollaries follow from Theorems~\ref{thm:main1} and~\ref{thm:main2} in rather standard ways, as suggested in the introduction. See Sections~\ref{sec:consequenceskCycles} and~\ref{subsec:FourCyclesConsequences} for the full details.

\paragraph{A road-map for the technical parts:} In Section~\ref{sec:RemoveMostCyles} we prove Theorem~\ref{thm:main1}, as well as its hardness consequences for distance oracles, dynamic APSP, and $k$-cycle enumeration. In Section~\ref{sec:RemoveAllFourCycles}, we prove Theorem~\ref{thm:main2}, as well as its hardness consequences for triangle finding in $4$-cycle free graphs, $4$-cycle finding, and girth approximation.



\section{Further Related Work}
Many previous works derive consequences from high-girth graphs for distance computation problems.
For example, in lower bounds for graph sparsification (spanners) \cite{PS89,Alt+93} or compression (distance oracles) \cite{bourgain1985lipschitz,matouvsek1996distortion,thorup2005approximate,sommer2009distance}, and for the number of rounds in a distributed setting \cite{GKP20,DruckerKO13} the lower bound constructions are built on constructions of a high-girth graph, either explicit (see \cite{hoory2002graphs}) or hypothetical under the Girth Conjecture~\cite{Erdos65}. 
Unfortunately, no one has managed to make such an approach work for conditional lower bounds in P, because such constructions are \emph{too structured} to be worst-case graphs and cannot encode a hard worst-case instance of another problem such as 3SUM.
Our approach is diametrically opposed: we start from a worst-case graph and \emph{efficiently} turn it into an (almost) high-girth graph by our short cycle removal technique (albeit with worse parameters than the best explicit constructions).

Countless papers in fine-grained complexity and graph algorithms study triangles and short cycles.
Let us mention a few that are more relevant for this work.
Roditty and Vassilevska Williams \cite{roditty2012subquadratic} proved a conditional lower bound for a particular approach for approximating the girth of a graph, but left proving hardness (for any algorithm) as a main open question.
Dahlgaard, Knudsen, and St{\"{o}}ckel \cite{DKS17} prove the hardness of $k$-cycle detection for all $k\geq 5$ assuming the hardness of the $k=3$ and $k=4$ cases.
Dudek and Gawrychowski prove that counting $4$-cycles is equivalent to computing the quartet distance on trees \cite{DG19} and to counting $4$-patterns in permutations \cite{DG20}.
Unlike in undirected graphs where the $k$-Cycle detection problem tends to become easier as $k$ grows (because the graph gets sparser), this is not the case in directed graphs where it is conjectured that $n^{2-o(1)}$ is required for large enough $k\geq 3$, and this conjecture is implied by other conjectures on the $k$-Clique problem~\cite{LincolnWW18,AgarwalR18}.

Besides the two general techniques mentioned above for hardness of approximation in P, and their applications, there are also results that follow from problem-specific tweaks to the exact-lower-bound constructions.
In the context of distance computations in graphs, some examples are for APSP \cite{DHZ00,Aingworthetal}, Diameter and related problems (without the use of graph-products) \cite{RV13,ChechikLRSTW14,AVW16}, for the Girth in directed graphs \cite{DV20} (see \cite{PachockiRSTW18,chechik2020constant} for recent upper bounds for this problem), and for dynamic near-additive spanners \cite{BergamaschiHGWW21}.
A more general result talks about the possibility of avoiding the $\log{W}$ factor that comes from the standard \emph{scaling} trick in $(1+\eps)$-approximations for a problem with weights up to $W$ by reduction to an unweighted problem~\cite{BringmannKW19}.
Moreover, there is a connection between \emph{deterministic} approximation algorithms and circuit complexity that has lead to strong inapproximability results~\cite{AB17approx,AR18,Chen+19} but the hardness assumptions underlying such barriers are known to be breakable with randomized algorithms.

Finally, we mention that we are not aware of previous papers studying the complexity of problems when the input is restricted to have a high girth, but such questions were already studied in the context of graph spanners~\cite{BHKP10,ABP18}.

%
%
%
%
%
%
%
%
%

\section{Preliminaries}

Most of our lower bounds rely on the following theorem, which establishes hardness of a triangle finding problem based on either of the standard conjectures about $3$-SUM or All Pairs Shortest Paths (for background on these conjectures see e.g.~\cite{Vass15,virginiaICM}).

\begin{theorem}[All-Edges-Triangle is 3-SUM and APSP hard~\cite{VX20}]\label{thm:3sumReduction}
Let $\eps>0$. Assuming the $3$-SUM conjecture or the APSP conjecture, no $O(n^{2-\eps})$-time algorithm can answer for each edge whether it participates in a triangle in a given $n$-node graph with maximum degree at most $\sqrt{n}$. 
\end{theorem}

For some background on this theorem, we mention that reductions from $3$-SUM to triangle listing were initiated by P{\u{a}}tra{\c{s}}cu~\cite{Pat10} and further refined in~\cite{KPP16}. Recently, Vassilevska Williams and Xu~\cite{VX20} showed that the $3$-SUM conjecture can be replaced by the APSP conjecture, obtaining the same result under either of these conjectures. They also showed a variant of this lower bound which replaces triangle listing by asking for every edge whether it is part of a triangle~\cite[Corollary 3.9]{VX20}; this variant is more useful in our context and stated above (slightly rephrased).

\section{Removing Most $k$-Cycles}\label{sec:RemoveMostCyles}

In this section we prove Theorem~\ref{thm:main1}.

\ThmRemovingCycles*

The proof of Theorem~\ref{thm:main1} is provided in Section~\ref{sec:hittingCycles}. We start with the dense-piece removal step, which is described in  Section~\ref{sec:DensePiece}. Then, in Section~\ref{sec:hittingCycles}, we use a randomized slicing together with the dense-piece removal step to deduce the theorem. In Section~\ref{sec:consequenceskCycles}, we prove hardness results that are  consequences of the theorem.

\subsection{Dense Piece Removal}\label{sec:DensePiece}


In this section we prove the following lemma. 

\begin{lemma}\label{lem:hardnessNoCycles}
Let $\gamma=(\omega-1)/4+\eps$ for an $\eps\in (0,\frac{3-\omega}{4})$, and let $k\geq 4$. Given a graph $G=(V,E)$ with maximum degree at most $\sqrt{n}$, there is an $O(n^{2-\eps})$-time algorithm that returns a subset of the edges $\tilde{E} \subseteq E$ and reports all the edges in $E$ that are in a triangle using an edge from $\tilde{E}$, such that the graph $\tilde{G}=(V,E\setminus \tilde{E})$ has at most $O(n^{k/2+\gamma})$ $k$-cycles.
\end{lemma}

That is, Lemma~\ref{lem:hardnessNoCycles} implies that in order to solve All-Edge Triangle in $G$ in subquadratic time, it suffices to solve All-Edge Triangle in the graph $\tilde{G}=(V,E\setminus \tilde{E})$ in subquadratic time. This is because after reporting all the edges in the graph that are in a triangle using an edge from $\tilde{E}$, it is safe to remove $\tilde{E}$ from the graph (without unintentionally removing triangles with unreported edges) and solve All-Edge Triangle in the obtained graph $\tilde{G}$. The advantage of  reducing the problem to solving All-Edge Triangle in $\tilde{G}$ is that $\tilde{G}$ is guaranteed to have significantly less $k$-cycles compared to a worst-case instance.

\paragraph{A road-map for the proof of Lemma~\ref{lem:hardnessNoCycles}:} We start with the useful definition of a $\gamma$-dense piece and a $\gamma$-dense path (Definition~\ref{def:densePiece}). In Lemma~\ref{lem:ManyDensePairs}, we show that a graph that has many $k$-cycles must contain many dense paths, a property that we use in Lemma~\ref{lem:find-dense-piece} to show that a dense subgraph can be found efficiently. In Lemma~\ref{lem:check-triangle-piece}, we show that given a $\sqrt{n}$-node subgraph, we can check for all edges $e$ in the graph, whether there is a triangle that uses $e$ and an edge from this subgraph efficiently. After the proof of Lemma~\ref{lem:check-triangle-piece}, we put everything together and present the formal proof of Lemma~\ref{lem:hardnessNoCycles}. Finally, in Theorem~\ref{thm:All-Edge-Few-Triangles}, we use the same ideas to show that All-Edge Triangle is 3SUM and APSP hard even when the graph contains a subquadratic number of triangles, a property that we need in one of our applications in Section~\ref{sec:consequenceskCycles} (namely, the lower bound for $k$-cycle enumeration).

\begin{definition}[\textbf{$\gamma$-Dense Pieces and $\gamma$-Dense Paths}]\label{def:densePiece}
Given an $n$-node graph with maximum degree at most $\sqrt{n}$, we say that a set of nodes of size at most $2\sqrt{n}$ is a $\gamma$-dense piece if the subgraph induced by these nodes has at least  $n^{1/2+\gamma}/2$ edges. Furthermore, we say that a simple $(k-2)$-path, $\{u,u_1,\cdots, u_{k-4},w\}$, is $\gamma$-dense if the number of $k$-cycles that use it is at least  $n^{1/2+\gamma}/2$.
\end{definition}

Hence, a $\gamma$-dense $(k-2)$-path, $\{u,u_1,\cdots, u_{k-4},w\}$, implies that there are $n^{1/2+\gamma}/2$ edges between $N(u)$ and $N(w)$, which implies that $N(u)\cup N(w)$ is a $\gamma$-dense piece. In the following lemma, we show that if there are many $k$-cycles in the graph then there are many $\gamma$-dense $(k-2)$-paths, which implies that there are many $\gamma$-dense pieces. 

\begin{lemma}\label{lem:ManyDensePairs}
Let $k\geq 4$. For every $0\leq \gamma< 1/2$, every $n$-node graph with maximum degree at most $\sqrt{n}$ that has at least $n^{k/2+\gamma}$ $k$-cycles must contain at least $\frac{1}{2}n^{k/2-1+\gamma}$ simple $\gamma$-dense $(k-2)$-paths $\{u,u_1,\cdots, u_{k-4},w\}$.
\begin{proof}
Let $c$ be the number of $k$-cycles in the graph, and for a simple path $p$ of $k-2$ nodes, let $c_p$ be the number of $k$-cycles that use $p$ as a subpath. Observe that 

$$c \leq \sum_{(k-2)\text{-paths } p} c_p$$

Furthermore, the number of $(k-2)$-paths in the graph is at most $n\cdot (\sqrt{n})^{k-3}=n^{\frac{k-1}{2}}$, because there are $n$ ways to pick the first node in the path, and $(\sqrt{n})^{k-3}$ ways to extend this node to a $(k-2)$-path. Moreover, observe that each $(k-2)$-path, $\{u,u_1,\cdots, u_{k-4},w\}$, participates in at most $n$ $k$-cycles, because there are at most $n$ edges between $N(u)$ and $N(w)$. 
Hence, since each $(k-2)$-path that is not $\gamma$-dense participates in at most $\frac{1}{2}n^{1/2+\gamma}$ $k$-cycles, if we have fewer than $\frac{1}{2}n^{k/2-1+\gamma}$ $\gamma$-dense $(k-2)$-paths, this implies that the number of $k$ cycles is bounded by 

$$c < \frac{1}{2}n\cdot n^{k/2-1+\gamma} + \frac{1}{2}n^{1/2+\gamma}\cdot n^{\frac{k-1}{2}}=n^{k/2+\gamma},$$
which is a contradiction. 
\end{proof}
\end{lemma}

The remainder of this section is devoted to showing that there is a subquadratic-time algorithm that removes all dense pieces, leaving the obtained graph with fewer than $n^{k/2+\gamma}$ $k$-cycles by Lemma~\ref{lem:ManyDensePairs}. We start with the following proposition that follows by a standard Chernoff argument.

\begin{proposition}\label{prop:EstimateEdges}
Let $X$ and $Y$ be two $\sqrt{n}$-size sets of nodes (not necessarily different). Sample $S=200n^{1/2-\gamma}\log n$ pairs in $X\times Y$ independently and uniformly at random. It holds that:
\begin{enumerate}
    \item If the number of edges between $X$ and $Y$ is at least $n^{1/2+\gamma}/2$, then at least $50\log n$ of the sampled pairs are edges, with probability at least $1-1/n^{10}$.
    \item If the number of edges between $X$ and $Y$ is smaller than $n^{1/2+\gamma}/200$, then fewer than $50\log n$ of the sampled pairs are edges, with probability at least $1-1/n^{10}$.
\end{enumerate}
\end{proposition}

In the following lemma, we show that if the graph has at least $n^{k/2+\gamma}$ $k$-cycles, then a dense subgraph can be found efficiently.

\begin{lemma}\label{lem:find-dense-piece}
Let $k\geq 4$. Given an $n$-node graph of maximum degree at most $\sqrt{n}$ that contains at least $n^{k/2+\gamma}$ $k$-cycles, there is an $O(n^{1-2\gamma}\log^2n)$-time algorithm that finds a pair of nodes $\{u,w\}$, such that the number of edge between $N(u)$ and $N(w)$ is at least $n^{1/2+\gamma}/200$, with high probability. 
\begin{proof}
First, we show that we can sample a $\gamma$-dense $(k-2)$-path efficiently. Observe that in $O(1)$-time, we can sample a $(k-2)$-path, such that each simple path  $\{u,u_1,\cdots, u_{k-4},w\}$ is sampled with probability at least $1/n^{(k-1)/2}$. This can be done by first sampling a starting node, and in each step we sample a node from the neighborhood of the previously sampled node, until we sample a $(k-2)$-path. Hence, by Lemma~\ref{lem:ManyDensePairs}, the probability mass of the simple $(k-2)$-paths that are $\gamma$-dense is at least $n^{\frac{k}{2}-1+\gamma}/n^{(k-1)/2}=n^{\gamma-1/2}$. Therefore, by sampling $100n^{1/2-\gamma}\log n$ such $(k-2)$-paths, one of them is a simple $\gamma$-dense path with high probability. Moreover, for each sampled $(k-2)$-path, $\{u,u_1,\cdots, u_{k-4},w\}$, we sample $200n^{1/2-\gamma}\log n$ pairs in $N(u)\times N(w)$, and check how many of the sampled pairs are edges. If the number is at least $50\log n$, we output the pair of nodes $\{u,w\}$. By Proposition~\ref{prop:EstimateEdges}, the algorithm finds at least one pair $\{u,w\}$ with the desired property with high probability. Furthermore, by Proposition~\ref{prop:EstimateEdges}, for any pair $\{u,w\}$ that the algorithm outputs there are at least $n^{1/2+\gamma}/200$ edges between $N(u)$ and $N(w)$, with high probability.
\end{proof}
\end{lemma}

Next, we show that given two $\sqrt{n}$-size sets of nodes $X$ and $Y$ (not necessarily disjoint), there is an efficient algorithm that reports all the edges in the graph that are in a triangle using an edge from $X\times Y$.

\begin{lemma}\label{lem:check-triangle-piece}
Given an $n$-node graph with maximum degree at most $\sqrt{n}$, and two $\sqrt{n}$-size sets of nodes $X$ and $Y$, there is an $O(n^{\frac{3+\omega}{4}})$-time algorithm that finds all the edges in the graph that are in a triangle using an edge from $X\times Y$.
\begin{proof}
Let $\beta>0$ be a constant to be chosen later, $V^h$ be the set of nodes that have at least $n^{1/2-\beta}$ neighbors in $X\cup Y$, and $V^{\ell}$ be the set of nodes that have fewer than $n^{1/2-\beta}$ neighbors in $X\cup Y$. Observe that for any edge $e$ that is in a triangle that uses an edge from $X\times Y$, it holds that either the triangle uses a node from $V^{\ell}$ (in which case either $e\in V^\ell\times (X\cup Y)$ or $e\in X\times Y$ and the third node is in $V^{\ell}$), or it uses a node from $V^H$ (in which case either $e\in V^H\times (X\cup Y)$ or $e\in X\times Y$ and the third node is in $V^H$). We start with the low-degree case, in which we find all the triangles that use a node from $V^{\ell}$.

\paragraph{Low-degree nodes:}
First, observe that we can find the set of nodes $V^{\ell}$ in linear time, as follows. We go over the nodes in $X\cup Y$, and for each of them we mark its neighbors. Then, we go over all the nodes in the graph and take those that were marked fewer than $n^{1/2-\beta}$ times to $V^{\ell}$.

To find the triangles involving nodes in $V^{\ell}$, we go over all the nodes in $V^{\ell}$ and for each of them we go over all pairs of neighbors in $X\cup Y$, and for each such pair we check whether it is an edge. To analyse the time complexity of this step, we bucket the set of nodes in $V^{\ell}$ by their degrees in $X\cup Y$, where the $i$'th bucket contains every node $v \in V^{\ell}$ of degree $|N(v) \cap (X\cup Y)| \in [2^i,2^{i+1})$, for $0\leq i\leq \log (n^{1/2-\beta})$. Observe that the time it takes to process the $i$'th bucket is $O(\frac{n}{2^i}\cdot (2^{i+1})^2)$. This is because there are $O(n/2^i)$ nodes with degrees in $[2^i,2^{i+1})$, since the number of edges incident to $X \cup Y$ is $O(n)$. Hence, in total for all buckets, this takes time
$$O\Big(\sum_{i}\frac{n}{2^i}\cdot (2^{i+1})^2\Big)=O(n\cdot n^{1/2-\beta})=O(n^{3/2-\beta}).$$

\paragraph{High-degree nodes:} Observe that the set $V^h$ can be found in linear time in a similar way that we used to find the set $V^{\ell}$. Furthermore, the size of $V^h$ is at most $O(n^{1/2+\beta})$ since the total number of edges incident to nodes in $X\cup Y$ is at most $O(n)$. It remains to find all the edges $e$ that are in a triangle that uses a node from $V^h$ and an edge from $X\times Y$. Hence, either $e\in X\times Y$ (in which case we want to check whether it is in a triangle that uses a node from $V^h$), $e\in V^h\times X$ (in which case we want to check whether it is in a triangle that uses a node from $Y$), or $e\in V^h\times Y$ (in which case we want to check whether it is in a triangle that uses a node from $X$). To find these edges, we use a matrix multiplication approach, as follows. 

\paragraph{$e\in X\times Y$:} To find all the edges in $X\times Y$ that are in a triangle that uses a node from $V^h$, we use a matrix multiplication algorithm. Consider the Boolean matrices $X\times V^h$ and $V^h\times Y$, where 1-entries indicate edges. By multiplying the two matrices, we get all the pairs $u\in X, w\in Y$ for which there is a 2-path between $u$ and $w$ through $V^h$. Hence, by going over all the edges in $X\times Y$, we can check for each of them whether it participates in a triangle that uses a node from $V^h$. Multiplying an $n^{1/2}\times n^{1/2+{\beta}}$ matrix by an $n^{1/2+\beta}\times n^{1/2}$ matrix takes time $O(n^{w/2}\cdot n^{\beta})$, because we can split it into $n^{\beta}$ many matrix products on $n^{1/2} \times n^{1/2}$ square matrices.\footnote{This step could be improved using rectangular matrix multiplication, which for the current value of $\omega$ would yield better constants in our lower bounds. Since in the limit for $\omega = 2 + o(1)$ our lower bounds are unaffected, we omit the details.} Furthermore, going over all the edges in $X\times Y$ takes $O(n)$ time. Hence, finding all the edges in $X\times Y$ that are in a triangle that uses a node from $V^h$ takes time $O(n^{w/2+\beta})$. 

\paragraph{$e\in V^h\times X$ or $e\in V^h\times Y$:} To find the edges in $V^h\times X$ that are in a triangle that uses a node from $Y$, we use a similar matrix multiplication algorithm. This time, consider the matrices $V^h\times Y$ and $Y\times X$, where 1-entries indicate edges. By multiplying the two matrices, we get all the pairs $u\in V^h, x\in X$ for which there is a 2-path between $u$ and $x$ through $Y$. Hence, by going over all the edges in $V^h\times X$ (which takes $O(n^{1+\beta})$ time), we can check for each of them whether it participates in a triangle that uses a node from $Y$. Multiplying an $n^{1/2}\times n^{1/2+{\beta}}$ matrix by an $n^{1/2}\times n^{1/2}$ matrix takes $O(n^{w/2}\cdot n^{\beta})$ time as well. 

Finding the edges in in $V^h\times Y$ that are in a triangle that uses a node from $X$ is done in a similar way. Hence, in total, the high degree case takes $O(n^{\omega/2+\beta})$ time. 

\paragraph{Putting everything together}
To optimize the time complexity in total for the high-degree and the low-degree cases, we set $\beta=(3-\omega)/4$, which implies a total running time of $O(n^{(3+\omega)/4})$, as desired.
\end{proof}
\end{lemma}

Now we are ready to prove Lemma~\ref{lem:hardnessNoCycles}.

\begin{proof}[Proof of Lemma~\ref{lem:hardnessNoCycles}]
We iteratively run the following algorithm that has a 3-step structure: (1) Find a pair  $\{u,w\}$ with at least $n^{1/2+\gamma}/200$ edges between $N(u)$ and $N(w)$ by using the algorithm from Lemma~\ref{lem:find-dense-piece} (if the algorithm from Lemma~\ref{lem:find-dense-piece} fails to find such a pair, we know that the graph has fewer than $n^{k/2+\gamma}$ $k$-cycles, and we stop), (2) report all the edges that are in a triangle that uses an edge from $N(u)\times N(w)$ by using the algorithm from Lemma~\ref{lem:check-triangle-piece}, and (3) remove all the edges between $N(u)$ and $N(w)$ from the graph. Since the number of edges in the graph is at most $n^{3/2}$, and in each step we remove at least $n^{1/2+\gamma}/200$ edges, the algorithm has at most $O(n^{1-\gamma})$ iterations. Furthermore, in each iteration, step (1) takes $O(n^{1-2\gamma}\log^2n)$ time, step (2) takes $O(n^{(3+\omega)/4})$ time, and step (3) takes linear time. Hence, in total, the running time is $O(n^{1-\gamma}\cdot n^{(3+\omega)/4})$. For $\gamma=(\omega-1)/4+\eps$, where $\eps\in (0,\frac{3-\omega}{4})$, this is $O(n^{2-\eps})$, as desired (the upper bound on $\eps$ is needed so that we have $\gamma<1/2$).
\end{proof}

Finally, we finish this section with the following remark and theorem on the number of triangles in the All-Edge Triangle problem.

\paragraph{A remark on the All-Edge Triangle problem:} One of our lower bound results in Section~\ref{sec:consequenceskCycles} requires the total number of triangles in the All-Edge Triangle instance to be subquadratic, specifically the lower bound for $4$-cycle enumeration in Theorem~\ref{thm:4-cycle-enumeration}. Interestingly, we can use the same ideas that we presented in this section to show that the All-Edge Triangle problem is still (3-SUM and APSP) hard even when the number of triangles is $O(n^{3/2+\gamma})$, for $\gamma=\frac{\omega-1}{4}+\eps$.

\begin{theorem}\label{thm:All-Edge-Few-Triangles}
Let $\eps\in (0,\frac{3-\omega}{4})$ and $\gamma=\frac{\omega-1}{4}+\eps$. Assuming the $3$-SUM conjecture or the APSP conjecture, no $O(n^{2-\eps})$-time algorithm can answer for each edge whether it participates in a triangle in a given $n$-node graph with maximum degree at most $\sqrt{n}$, and $O(n^{3/2+\gamma})$ triangles. 
\begin{proof}
By Theorem~\ref{thm:3sumReduction}, assuming the $3$-SUM conjecture or the APSP conjecture, no $O(n^{2-\eps})$-time algorithm can answer for each edge whether it participates in a triangle in a given $n$-node graph with maximum degree at most $\sqrt{n}$. Hence, it suffices to show that we can reduce the number of triangles in the input graph to $n^{3/2+\gamma}$ in subquadratic time. 

For this, we use the same dense piece removal trick that we introduced in this section. Observe that an $n$-node graph with at least $n^{3/2+\gamma}$ triangles must contain at least $\frac{1}{2}n^{1/2+\gamma}$ nodes that each participate in at least $\frac{1}{2}n^{1/2+\gamma}$ triangles. Otherwise, since each node can be in at most $n$ triangles, the number of triangles would be smaller than $\frac{1}{2}n\cdot n^{1/2+\gamma}+\frac{1}{2}n^{1/2+\gamma}\cdot n=n^{3/2+\gamma}$, which is a contradiction. Hence, there are at least $\frac{1}{2}n^{1/2+\gamma}$ nodes $v$ for which the number of edges between the nodes in $N(v)$ is at least $\frac{1}{2}n^{1/2+\gamma}$.

Therefore, as long as the number of triangles is at least $n^{3/2+\gamma}$, we can find a $\sqrt{n}$-node subgraph with $\Omega(n^{1/2+\gamma})$ edges efficiently: sample $100n^{1/2-\gamma}\log n$ nodes $v$, and for each of them sample $200n^{1/2-\gamma}\log n$ pairs in $N(v)\times N(v)$ and check how many of the sampled pairs are edges. If the number of edges is at least $50\log n$, we output $N(v)$. By a similar analysis to the one in Lemma~\ref{lem:find-dense-piece}, as long as long as the number of triangles is at least $n^{3/2+\gamma}$, this algorithm finds a $\sqrt{n}$-node subgraph with $\Omega(n^{1/2+\gamma})$ edges, with high probability. Moreover, by Lemma~\ref{lem:check-triangle-piece}, we report all the edges in the graph that participate in a triangle that uses an edge from the subgraph in  $O(n^{\frac{3+\omega}{4}})$ time. 

Therefore, by iteratively finding a dense subgraph, checking for each edge in the subgraph whether there is a triangle that uses it, and removing these edges from the graph, we obtain a graph with fewer than $n^{3/2+\gamma}$ triangles. By a similar analysis to the one provided in the proof of Lemma~\ref{lem:hardnessNoCycles}, this takes subquadratic time, as desired.
\end{proof}
\end{theorem}

\subsection{Hitting $k$-Cycles Faster than Triangles: A Proof of Theorem~\ref{thm:main1}}\label{sec:hittingCycles}

\begin{proof}[Proof of Theorem~\ref{thm:main1}]
Let $\gamma=(\omega-1)/4+\eps$, where $\eps\in (0,\frac{3-\omega}{4})$. Recall that we can assume without loss of generality that the input graph is tripartite. Let $A,B$, and $C$ be the three parts. First, we run the algorithm from the dense piece removal step (Lemma~\ref{lem:hardnessNoCycles}) for every $4\leq k'\leq k$, and we set $E'$ to be the set of reported edges that are in a triangle that uses an edge from a dense piece $\tilde{E}$. This takes $O(kn^{2-\eps})=O(n^{2-\eps})$ time. Furthermore, the obtained graph has fewer than $n^{k'/2+\gamma}$ $k'$-cycles for every $4\leq k'\leq k$.

We break the obtained graph into tripartite subgraphs, which we refer to as slices, such that each triangle appears in exactly one slice, as follows. We randomly partition each of the sets $A, B$ and $C$ into $n^{1/2-\alpha}$ sets, each of expected size $n^{1/2+\alpha}$, where each node joins each of the sets uniformly at random, and independently of the choices for the other nodes. We denote these sets by $\{A_j\}_{j\in [n^{1/2-\alpha}]}$, $\{B_k\}_{k\in [n^{1/2-\alpha}]}$, and $\{C_{\ell}\}_{\ell\in [n^{1/2-\alpha}]}$. That is\footnote{We use the notation $\dot\cup$ to denote a disjoint union of sets.}, $$A=A_1 \dot\cup \ldots \dot\cup A_{n^{1/2-\alpha}}, \;\; B=B_1 \dot\cup \ldots \dot\cup B_{n^{1/2-\alpha}},\;\; C=C_1 \dot\cup \ldots \dot\cup C_{n^{1/2-\alpha}}.$$

By a standard Chernoff argument, for every slice $(A_j,B_k,C_{\ell})$, the number of nodes is $O(n^{1/2+\alpha})$ and the maximum degree is $O(n^{1/2} / n^{1/2-\alpha}) = O(n^{\alpha})$ with high probability. It remains to show that the expected number of $\{4,..,k\}$-cycles in each slice is $O(n^{\gamma+\alpha k})$. Observe that the probability that a given $k'$-cycle is fully contained in the slice $(A_j,B_k,C_{\ell})$ is $1/(n^{1/2-\alpha})^{k'
}=1/(n^{k'/2-k'\alpha})$. Hence, the expected number of $k'$-cycles that are fully contained in the slice $(A_j,B_k,C_{\ell})$ is $n^{k'/2+\gamma}/(n^{k'/2-k\alpha})=n^{\gamma+k'\alpha}$. Over all $4\leq k'\leq k$, the expected number of $\{4,..,k\}$-cycles is at most $kn^{\gamma+k\alpha}=O(n^{\gamma+k\alpha})$ (since $k$ is constant), as desired.
\end{proof}

\subsection{Consequences of Theorem~\ref{thm:main1}}\label{sec:consequenceskCycles}

We start with a hardness result for $4$-cycle enumeration.

\begin{theorem}\label{thm:4-cycle-enumeration}
For every $\eps>0$ there is a $\delta>0$ such that if there is an algorithm that can preprocess an $m$-edge graph in $O(m^{1+\frac{3-\omega}{2(4-\omega)}-\eps})$ time and then enumerate $4$-cycles with $m^{o(1)}$ delay, then there is an $n^{2-\delta+o(1)}$-time algorithm that given an $n$-node graph with maximum degree at most $\sqrt{n}$ and $O(n^{2-\delta})$ triangles answers for every edge whether it participates in a triangle with probability at least $9/10$. 
\end{theorem}
\begin{proof}
 Let $G$ be an $n$-node tripartite graph with sides $A,B,$ and $C$. First, we run the subquadratic-time algorithm from Theorem~\ref{thm:main1} with $\eps',\alpha$ to be chosen later, and $k=4$. Since the expected number of $4$-cycles in each $G_i$ is $O(n^{\frac{\omega-1}{4}+\eps'+4\alpha})$, it follows that the total number of $4$-cycles in all the $G_i$'s is at most $O(n^{3/2+\frac{\omega-1}{4}+\alpha+\eps'})$  with probability at least $99/100$ (by linearity of expectation and Markov). Furthermore, since each edge that is in a triangle is either in $E'$ or is in a triangle in one of the $G_i$'s, in order to find the remaining edges in $E\setminus E'$ that are in a triangle, it suffices to enumerate the triangles in all the $G_i$'s.

For this, we subdivide the edges in $B\times C$ by adding a dummy node $bc$ on each edge $\{b,c\}$. Hence, any triangle $\{a,b,c\}$ in some $G_i$ becomes a $4$-cycle $\{a,b,bc,c\}$; these are the only newly introduced 4-cycles.  We refer to the other $4$-cycles that are not a result of subdividing triangles as \emph{false} $4$-cycles; note that each false 4-cycle already was a 4-cycle before subdividing the edges. For each $G_i$, we run a $4$-cycles enumeration algorithm on the subdivided graph. Let $P(m)=m^{x}$ be the preprocessing time and $D(m)=m^{o(1)}$ be the delay. Since each $G_i$ has $n^{1/2+2\alpha}$ edges with high probability, the total preprocessing time for all $G_i$'s is

$$O(n^{3/2-3\alpha}\cdot P(n^{1/2+2\alpha}))=O(n^{3/2-3\alpha+x(1/2+2\alpha)}),$$ 
with high probability. For $x<1+\alpha/(1/2+2\alpha)$ this is subquadratic. On the other hand, the total delay we spend on false $4$-cycles is $O(n^{3/2+\frac{\omega-1}{4}+\alpha+\eps'+o(1)})$ with probability at least $99/100$. The remaining delay is spent on enumerating subdivided triangles, and there is only a subquadratic number of them. Hence, for $\alpha=\frac{3-\omega}{4}-\eps''$, for $\eps''>\eps'$, the total delay is subquadratic. Furthermore, since

\begin{align*}
    1+\alpha/(1/2+2\alpha)&=1+\frac{\frac{3-\omega}{4}-\eps''}{1/2+2(\frac{3-\omega}{4}-\eps'')}\\&>1+\frac{\frac{3-\omega-4\eps''}{4}}{\frac{4-\omega}{2}}
    \\&>1+(3-\omega)/(2(4-\omega))-2\eps'',
\end{align*}
when we set $x=1+(3-\omega)/(2(4-\omega))-\eps$, for $\eps=2\eps''$, the total preprocessing time and the total delay are both subquadratic, as desired.
\end{proof}

Corollary~\ref{cor:fourEnum} follows immediately by combining Theorem~\ref{thm:4-cycle-enumeration} and Theorem~\ref{thm:All-Edge-Few-Triangles}.

\begin{corollary}\label{cor:fourEnum}

Assuming either the $3$-SUM or APSP Conjectures, no algorithm can process an $m$-edge graph in $O(m^{1+\frac{3-\omega}{2(4-\omega)}-\eps})$ time and then enumerate $4$-cycles with $m^{o(1)}$ delay.
\end{corollary}

By simple gadget reductions that show that $k$-cycle (detection, enumeration, or listing) for any $k$ is at least as hard as either the $k=3$ or $k=4$ case (see Appendix~\ref{sec:reduction34cycle}) the following corollary follows: 

\CorEnum*

Next, we show a hardness result for approximate distance oracles.
\begin{theorem}\label{thm:distanceOracles}
For every $\eps,\delta'>0$ there is a $\delta>0$ such that if there is an $O(m^{1+\frac{3-\omega}{2(k+1-\omega)}-\eps})$-time algorithm that can $(k/2-\delta')$-approximate the distances between $m$ given pairs of nodes in a given $m$-edge graph, then there is an $O(n^{2-\delta})$-time algorithm for $n$-node graphs with maximum degree at most $\sqrt{n}$ that with probability at least $9/10$ answers for every edge whether it participates in a triangle.
\end{theorem}
\begin{proof}
First, we run the subquadratic-time algorithm from Theorem~\ref{thm:main1} with $\eps'$ and $\alpha$ to be chosen later. For each of the returned $G_i$'s, we show how to check for every edge in $B\times C$ whether it participates in a triangle. (Checking the edges in $A\times B\cup A\times C$ is symmetric.) For every $G_i$, we remove the edges in $(B\times C)\cap E(G_i)$, and we denote the obtained graph by $G_i'$. We run a $(k/2-\delta')$-approximate distance oracle algorithm on $G_i'$, where we query all the $\{b,c\}$ pairs that correspond to the removed edges. We refer to the pairs $\{b,c\}$ for which the algorithm returned an estimate that is smaller than $k$ as the candidates of $G_i$. For each such candidate pair, we check whether the corresponding edge is in a triangle in $G_i$, which takes $O(n^{\alpha})$ time per edge. If the edge is found to be in a triangle, we remove it from the set of candidates of $G_j$ for every $j>i$. This ensures that we don't spend too much time on checking whether the same edge is in many different triangles.

Observe that for every edge $\{b,c\}$ that is in a triangle in $G_i$, the $(k/2-\delta')$-approximation algorithm must return an estimate that is smaller than $k$ when we query the pair $\{b,c\}$, as there is a two-path between $b$ and $c$ in $G_i'$. Furthermore, for every pair $\{b,c\}$ for which the algorithm returns an estimate that is smaller than $k$, it holds that there is a path between $b$ and $c$ of length at most $k-1$ in $G_i'$, and therefore the edge $\{b,c\}$ is in a cycle of length at most $k$ in $G_i$. We refer to the edges $\{b,c\}$ for which the algorithm returns an estimate $<k$ but $\{b,c\}$ is not in a triangle in $G_i$ as \emph{false} edges.

\paragraph{Running time:} Let $T(m)=m^x$ be the running time of the distance oracle algorithm (specifically, the total time for preprocessing an $m$-edge graph and answering $m$ approximate distance queries that are given in advance). In total, running this algorithm for all the $G_i'$'s takes time

$$O(n^{3/2-3\alpha+x(1/2+2\alpha)}).$$

For $x<1+\alpha/(1/2+2\alpha)$ this is subquadratic. Furthermore, the total number of $\{4,..,k\}$-cycles in all the $G_i$'s is $O(n^{3/2-3\alpha+(\omega-1)/4+\eps'+\alpha k})$ with probability at least $99/100$. Therefore, this is also the total number of times we check whether a false edge is in a triangle, over all the $G_i$'s. For an edge $\{b,c\}$ that participates in a triangle, we run a single check - the first time it was found to be in a triangle in some $G_i$. Hence, the total running time for this step is $O(n^{3/2+(\omega-1)/4+\eps'+\alpha (k-3)}\cdot n^{\alpha})=O(n^{3/2+(\omega-1)/4+\eps'+\alpha (k-2)})$. This is subquadratic when we set  

$$\alpha=\frac{\frac{1}{2}-\frac{\omega-1}{4}}{k-2}-\eps''=\frac{3-\omega}{4(k-2)}-\eps''$$
for $\eps''>\eps'$. For this choice of $\alpha$, we have that 

\begin{align*}
    1+\alpha/(1/2+2\alpha) &= 1 + \frac{\frac{3-\omega}{4(k-2)}-\eps''}{\frac{1}{2}+2(\frac{3-\omega}{4(k-2)}-\eps'')}\\&>1 + \frac{\frac{3-\omega-4(k-2)\eps''}{4(k-2)}}{\frac{k+1-\omega}{2(k-2)}}\\
    &=1 + \frac{3-\omega}{2(k+1-\omega)}-\frac{4(k-2)\eps''}{2(k+1-\omega)}\\
    &>1 + \frac{3-\omega}{2(k+1-\omega)}-\frac{4(k-2)\eps''}{2(k+1-3)}\\
    &=1 + \frac{3-\omega}{2(k+1-\omega)}-2\eps''
\end{align*}

Thus, when we set $x=1+\frac{3-\omega}{2(k+1-\omega)}-\eps$, for $\eps>2\eps''$, the total running time is subquadratic, as desired.
\end{proof}

Corollary~\ref{cor:distOracles} follows immediately by combining Theorem~\ref{thm:distanceOracles} with Theorem~\ref{thm:3sumReduction}.

\CorDistOracles*

Finally, we prove a hardness result for dynamic approximate All Pairs Shortest Paths.

\begin{theorem}\label{thm:dynamicAPSP}
For every $\eps,\delta'>0$ and integer $k \ge 4$ there is a $\delta>0$ such that if there is a dynamic algorithm for $(k/2-\delta')$-approximate APSP with preprocessing time $O(N^3)$ and update/query time $O(m^{\frac{3-\omega}{2(k+1-\omega)}-\eps})$ in $N$-node and $m$-edge graphs, then there is an $O(n^{2-\delta})$-time algorithm for $n$-node graphs with maximum degree at most $\sqrt{n}$ that with probability at least $9/10$ answers for every edge whether it participates in a triangle.
\end{theorem}
\begin{proof}
First, we run the subquadratic-time algorithm from Theorem~\ref{thm:main1} with $\eps'$ and $\alpha$ to be chosen later. We show how to use a dynamic algorithm to check for each edge $\{b,c\}\in B\times C$ whether it participates in a triangle. (Checking the edges in $A\times B\cup A\times C$ is symmetric.) For each $G_i$, we remove the $B\times C$ edges, obtaining a graph $G_i'$. We let $G_1'$ be the input graph to be preprocessed by the dynamic algorithm in $O((n^{(1/2+\alpha)})^3)=O(n^{3/2+3\alpha})$ time, and we consider the following sequence of updates and queries we feed into the dynamic algorithm. 

\paragraph{Sequence of updates and queries:} For $1\leq i\leq n^{3/2-3\alpha}$ phases, in each phase $i$ we make the following queries and updates. Queries: For each edge $\{b,c\}\in (B\times C)\cap E(G_i)$, we query the pair $\{b,c\}$. This takes $O(n^{1/2+2\alpha})$ queries. Updates: we delete all the edges in $G_i'$ and add all the edges in $G'_{i+1}$, by using $O(n^{1/2+2\alpha})$ updates. 

\paragraph{Postprocessing:}
We use the distance estimations returned by the queries to find for each edge $\{b,c\}$ in each $G_i$ whether it in a triangle in $G_i$, as follows. For each $G_i'$, we collect all the pairs $\{b,c\}$ for which the answer to the query is $< k$, and we refer to the corresponding edges as the candidates of $G_i$. For each candidate edge, we check whether it is in a triangle in $G_i$ by iterating over all neighbors of the endpoints, and if so, we remove the edge from the set of candidates of $G_j$ for every $j>i$. This finishes the reduction. 

\paragraph{Running time:} In total, for all phases, the number of queries and updates is $O(n^{2-\alpha})$. Hence, if each update and query takes time $O(n^{x(1/2+2\alpha)})$, we have that the total query and update time is $O(n^{x(1/2+2\alpha)}\cdot n^{2-\alpha})$, for all updates and queries.
For $x<\frac{\alpha}{1/2+2\alpha}$, this is subquadratic.

For the postprocessing the analysis is similar to the one in Theorem~\ref{thm:distanceOracles}: Checking whether a candidate edge forms a triangle takes time $O(n^{\alpha})$, and there are $O(n^{3/2+(\omega-1)/4+\eps'+\alpha (k-3)})$ candidate edges in total, so the postprocessing takes total time $O(n^{3/2+(\omega-1)/4+\eps'+\alpha (k-2)})$. For $\alpha=\frac{3-\omega}{4(k-2)}-\eps''$, for $\eps''>\eps'$, this is subquadratic. Since $\omega \ge 2$ and $k \ge 4$, this choice of $\alpha$ satisfies $\alpha < 1/6$, so also the preprocessing time of $O(n^{3/2+3\alpha})$ is subquadratic.


Furthermore, by a similar calculation to the one provided in Theorem~\ref{thm:distanceOracles}, we have that $\frac{\alpha}{1/2+2\alpha}>\frac{3-\omega}{2(k+1-\omega)}-2\eps''$. For $\eps>2\eps''$, we set $x=\frac{3-\omega}{2(k+1-\omega)}-\eps<\frac{\alpha}{1/2+2\alpha}$. Hence, since the number of edges at any time of the dynamic process is $m=O(n^{1/2+2\alpha})$, as a function of the number of edges $m$, if the update time is  $O(m^x)=O(m^{\frac{3-\omega}{2(k+1-\omega)}-\eps})$, the total running time of the above algorithm is subquadratic, as desired.
\end{proof}

The following corollary follows immediately by combining Theorem~\ref{thm:distanceOracles}, Theorem~\ref{thm:dynamicAPSP}, and Theorem~\ref{thm:3sumReduction}, where the first bullet follows by a straightforward reduction from the Offline Distance Oracles problem to Decremental Dynamic APSP (just preprocess the graph and answer the queries without ever making edge deletions).

\CorDynamic*

\section{Removing All $4$-Cycles}\label{sec:det}\label{sec:RemoveAllFourCycles}

 In this section we prove Theorem~\ref{thm:main2} (Section~\ref{subsec:RemoveAllFourCycles}), as well as some hardness consequences of it (Section~\ref{subsec:FourCyclesConsequences}).

\subsection{A proof of Theorem~\ref{thm:main2}}\label{subsec:RemoveAllFourCycles}

\ThmAllFourCycles*

\begin{proof}
Let $\alpha\in (0,\frac{3-\omega}{8})$ and $\eps,\eps'\in (0,\frac{3-\omega}{4}-2\alpha)$ to be chosen later. First, we run the subquadratic-time algorithm from Theorem~\ref{thm:main1} with $\eps',\alpha$, and $k=4$. Recall that this algorithm returns a set of edges $E'$ and $n^{3/2-3\alpha}$ subgraphs, such that each edge that is in a triangle is either in $E'$ or in a triangle in one of the subgraphs, where each subgraph is a slice $(A_j,B_k,C_{\ell})$, for $j,k,{\ell}\in [n^{1/2-\alpha}]$. Furthermore, each slice has $O(n^{(\omega-1)/4+\eps'+4\alpha})$ $4$-cycles in expectation, and therefore the overall number of $4$-cycles in all the slices is at most $O(n^{3/2+\alpha+(\omega-1)/4+\eps'})$ with probability at least $99/100$. Our algorithm adds more edges to $E'$ such that the obtained slices are $4$-cycle-free, as follows.

We show that it is possible to list all the $4$-cycles in all the slices in time that is linear in their number, see Lemma~\ref{lem:listfourcycles} below. After listing all the $4$-cycles, we denote by $\mathcal{S}_{j,k,{\ell}}$ the set of edges that participate in 4-cycles in the slice $(A_j,B_k,C_{\ell})$. Note that after removing the edges $\mathcal{S}_{j,k,{\ell}}$ from the slice $(A_j,B_k,C_{\ell})$ it becomes 4-cycle-free, as desired. It remains to check for each edge in $\mathcal{S}_{j,k,{\ell}}$ whether it participates in a triangle in the slice, and if so we add it to $E'$. Since the degree of each node in a slice is with high probability at most $O(n^{\alpha})$, this takes $O(|\mathcal{S}_{j,k,{\ell}}|\cdot n^{\alpha})$ time per slice with high probability. Hence, in total, for all slices, this takes time $\sum_{j,k,{\ell}} O(|\mathcal{S}_{j,k,{\ell}}|\cdot n^{\alpha})=O(n^{3/2+2\alpha+(\omega-1)/4+\eps'})$ with constant probability. By setting $\alpha=(3-\omega)/8-(\eps+\eps')/2$, this takes time $O(n^{2-\eps})$, as desired. It remains to show how to efficiently list all the $4$-cycles in all the slices in time that is linear in their number: 

\begin{lemma}\label{lem:listfourcycles}
We can enumerate all 4-cycles in any of the slices $(A_j,B_k,C_{\ell})$, for all $j,k,{\ell}\in [n^{1/2-\alpha}]$,
in total time $O(n^{3/2+\alpha} + c)$ 
where $c$ is the output size, that is, $c$ is the total number of such 4-cycles.
\begin{proof}

Observe that for any slice, any $4$-cycle uses exactly two nodes from one of the sides of the slice. In the following, we show how to list all the $4$-cycles that use two nodes from $A_j$, for all the slices $(A_j,B_k,C_{\ell})$. Listing all the $4$-cycles that use two nodes from $B_k$ or two nodes from $C_{\ell}$ is symmetric.  

We start with the following useful notations. For each slice $(A_j,B_k,C_{\ell})$, we think about $A_j$ as being the center of the slice, $B_k$ being the left side of the slice, and $C_{\ell}$ being the right side of the slice. For a set $A_j$ and a pair $\{a,a'\}\subseteq A_j$, let $L^{a,a'}_{k}$ be the set of nodes $b\in B_k$ for which $\{a,b,a'\}$ is a two-edge path. Similarly, let $R^{a,a'}_{\ell}$ be the set of nodes $c\in C_{\ell}$ for which $\{a,c,a'\}$ is a 2-path. Furthermore, for a set $A_j$ and a pair $\{a,a'\}\subseteq A_j$, let $I^{a,a'}_L$ be the set of coordinates $k\in [n^{1/2-\alpha}]$ for which $L^{a,a'}_k$ is not empty. Similarly, $I^{a,a'}_R$ is the set of coordinates ${\ell}\in [n^{1/2-\alpha}]$ for which $R^{a,a'}_{\ell}$ is not empty. Finally, for every pair of sets $A_j,B_k$, let $P^L_{j,k}$ be the set of pairs $\{a,a'\}\subseteq A_j$ for which $|L^{a,a'}_k|\geq 2$. Similarly, $P^R_{j,\ell}$ is the set of pairs $\{a,a'\}\subseteq A_j$ for which $|R^{a,a'}_{\ell}|\geq 2$. 

Our algorithm has a preprocessing step that computes all the sets $I^{a,a'}_L, I^{a,a'}_R$, all the nonempty sets $L^{a,a'}_k,R^{a,a'}_{\ell}$, and all the sets $P^L_{j,k}, P^R_{j,\ell}$ (for every $j,k,{\ell}\in [n^{1/2-\alpha}]$ and every pair $\{a,a'\}\in A_j$). 
Then we show that given the sets $I^{a,a'}_L, I^{a,a'}_R,L^{a,a'}_k,R^{a,a'}_{\ell}$ we can list all the 4-cycles that use two nodes from $A_j$ and one node from each of $B_k$ and $C_{\ell}$, and given the sets $P^L_{j,k},P^R_{j,\ell},L^{a,a'}_k,R^{a,a'}_{\ell}$ we can list all 4-cycles between every pair $A_j,B_k$, and every pair $A_j,C_{\ell}$. The details follow.

\paragraph{Preprocessing step:} 
Recall that we denote by $N(u)$ the set of neighbors of a node $u$ in the original graph. For each pair $A_j,B_k$, we go over all the nodes $b\in B_k$, and for each such node, we go over all the pairs $\{a,a'\}\subseteq N(b)\cap A_j$, and we add $k$ to $I^{a,a'}_L$ and $b$ to $L^{a,a'}_k$. If $|L^{a,a'}_k|\geq 2$, then we also add $\{a,a'\}$ to  $P^L_{j,k}$. Since the size of $B_k$ is $O(n^{1/2+\alpha})$, and the maximum degree of a node $b\in B_k$ in $A_j$ is $O(n^{\alpha})$, for a pair $A_j,B_k$, this takes time $O(n^{1/2+\alpha}\cdot n^{2\alpha})=O(n^{1/2+3\alpha})$. Hence, in total, for all pairs $A_j,B_k$, this takes time $O(n^{1-2\alpha}\cdot n^{1/2+3\alpha})=O(n^{3/2+\alpha})$. 
The sets $I^{a,a'}_R$, the nonempty sets $R^{a,a'}_{\ell}$, and the sets $P^R_{j,\ell}$ are computed symmetrically, for all $j,\ell$, and $\{a,a'\}\subseteq A_j$. 

\paragraph{Listing all 4-cycles between all pairs $A_j, B_k$ and all pairs $A_j,C_{\ell}$:} We show how to list all 4-cycles between all pairs $A_j,B_k$. Listing all the 4-cycles between all pairs $A_j,C_{\ell}$ is symmetric. Observe that the total number of 4-cycles between all pairs $A_j,B_k$ is

\begin{align*}
\sum_{j,k\in [n^{1/2-\alpha}]} \sum_{\{a,a'\}\in P^L_{j,k}} {|L^{a,a'}_j|\choose 2}
\end{align*}

 To list them, we go over all pairs $A_j,B_k$, and for each such pair we list all tuples $(a,b,a',b')$, where $\{a,a'\}\in P^L_{j,k}$, and $\{b,b'\}\subseteq L^{a,a'}_k$. Since  all the sets $P^L_{j,k}$ and $L^{a,a'}_k$ were already computed in the preprocessing step, this takes an amount of time which is linear in the number of 4-cycles. We list all the 4-cycles between all pairs $A_j,C_{\ell}$ in a similar way. 

\paragraph{Listing all 4-cycles that use two nodes from $A_j$ and one node from each of $B_k,C_{\ell}$, for every $j,k,\ell$:} Observe that the number of such 4-cycles is

\begin{align}\label{eq:4-cycles-tri}
    \sum_{j\in [n^{1/2-\alpha}]}\sum_{\{a,a'\}\subseteq A_j}\sum_{(k,\ell)\in I^{a,a'}_L\times I^{a,a'}_R}|L^{a,a'}_k|\cdot |R^{a,a'}_\ell|
\end{align}

 Our goal is to list all these 4-cycles in an amount of time that is linear in their number. For this, we go over all sets $A_j$, and for each such set, we go over all pairs $\{a,a'\}\subseteq A_j$, and for each such pair, we go over all pairs $(k,\ell)\in I^{a,a'}_L\times I^{a,a'}_R$, and list all the tuples $(a,b,a',c)$, where $b\in L^{a,a'}_k$ and $c\in R^{a,a'}_{\ell}$. The amount of time for this step is proportional to

\begin{align*}
    \sum_{j\in [n^{1/2-\alpha}]}\sum_{\{a,a'\}\subseteq A_j} \Bigg(1+\sum_{(k,\ell)\in I^{a,a'}_L\times I^{a,a'}_R}|L^{a,a'}_k|\cdot |R^{a,a'}_\ell| \Bigg).
\end{align*}

This is because for the pairs $\{a,a'\}$ that don't contribute any 4-cycle to the sum (\ref{eq:4-cycles-tri}) we spend constant time. For the other pairs, the amount of time we spend is proportional to the number of 4-cycles they participate in. 
Note that the summand 1 contributes $O(n^{1/2-\alpha} \cdot |A_j|^2) = O(n^{1/2-\alpha} \cdot (n^{1/2+\alpha})^2) = O(n^{3/2+\alpha})$ to the running time. The other summand is simply the total number of 4-cycles as in (\ref{eq:4-cycles-tri}). We thus obtain total time $O(n^{3/2+\alpha} + c)$, as desired.
\end{proof}
\end{lemma}
This finishes the proof of Theorem~\ref{thm:main2}.
\end{proof}

\subsection{Consequences of Theorem~\ref{thm:main2}}\label{subsec:FourCyclesConsequences}

We start with a reduction from triangle detection to $4$-cycle detection.

\begin{theorem}\label{lem:TimeFourDet}
For every $\delta>0$ there is a $\delta'>0$ such that if there is an $O(m^{1+\frac{3-\omega}{2(5-\omega)}-\delta})$-time algorithm for 4-cycle detection, then there is an $O(n^{2-\delta'})$-time algorithm for triangle detection in $n$-node graphs with maximum degree at most $\sqrt{n}$. 
\begin{proof}
First, we run the subquadratic-time algorithm from Theorem~\ref{thm:main2} with an arbitrarily small constant $\eps>0$ and $\alpha<(3-\omega)/8-\eps/2$ to be chosen later. Since the algorithm already checked for each edge in $E'$ whether it participates in a triangle, and since each triangle either uses an edge from $E'$ or is in one of the $G_i$'s, it remains to solve triangle detection in each $G_i$.

For this, we add a dummy node $bc$ on each edge $\{b,c\}\in B\times C$, which converts any triangle $\{a,b,c\}$ to a $4$-cycle $\{a,b,bc,c\}$. Furthermore, since none of the $G_i$'s had a $4$-cycle before adding the dummy nodes, the existence of a $4$-cycle in $G_i$ implies the existence of a triangle. Therefore, to solve triangle detection, it suffices to run a $4$-cycle detection algorithm in all the obtained $G_i$'s. Let $T(m)=m^x$ be the time complexity for $4$-cycle detection in $m$-edge graphs. Since each $G_i$ has $O(n^{1/2+2\alpha})$ edges with high probability, the total running time for all the $G_i$'s is  $$O(n^{3/2-3\alpha}\cdot(n^{1/2+2\alpha})^x)=O(n^{\frac{3}{2}-3\alpha+x(1/2+2\alpha)})$$

For $x<1+\alpha/(1/2+2\alpha)$, this is subquadratic. Hence, by setting $\alpha=(3-\omega)/8-\eps'$, for some $\eps'>\eps/2$, the running time from Theorem~\ref{thm:main1} is subquadratic, and

\begin{align*}
    1+\alpha/(1/2+2\alpha)&=1+\frac{\frac{3-\omega}{8}-\eps'}{1/2+2(\frac{3-\omega}{8}-\eps')}\\
    &>1+\frac{3-\omega}{2(5-\omega)}-\frac{8\eps'}{2(5-\omega)}
\end{align*}

By setting $x=1+\frac{3-\omega}{2(5-\omega)}-\delta$, for $\delta>8\eps'$, the total running time for $4$-cycle detection in all the $G_i$'s is subquadratic, as desired.
\end{proof}
\end{theorem}

Corollary~\ref{cor:fourFree} follows immediately from Theorem~\ref{lem:TimeFourDet}. The second and third bullets follow by essentially the same proof as the one provided for Theorem~\ref{lem:TimeFourDet}. Instead of running a $4$-Cycle detection algorithm, we run a triangle detection in $4$-cycle free graphs for the second bullet, and a girth approximation algorithm for the third bullet.

\CorFourFree*



\section{On the Hardness of Triangle}
\label{sec:discussion}

The conditional lower bounds in this paper are based on the $n^{2-o(1)}$ time hardness of two versions of triangle finding in $\sqrt{n}$-degree graphs:
The \emph{all-edge} version of reporting for each of the $n^{1.5}$ edges whether it is in a triangle, and the more basic \emph{detection} version of just deciding if there is any triangle in the graph.
The former is already known to be hard under either the 3SUM or APSP conjectures~\cite{VX20}, two of the most central conjectures in fine-grained complexity~\cite{Vass15,virginiaICM}, and therefore does not need further justification (see also \cite{DurajK0W20} for equivalences to range reporting problems). 
The goal of this section is to discuss the latter assumption.

Abboud and Vassilevska Williams \cite{AV14} introduced the following Triangle Conjecture and used it to prove hardness result for dynamic problems; the conjecture has also been used elsewhere, e.g. in databases \cite{CK19}.

\begin{conjecture}[The Triangle Conjecture \cite{AV14}]
Triangle detection requires $m^{4/3-o(1)}$ time, for some density regime $m=\Omega(n)$. In other words, there exists a constant $1 \le \alpha \le 2$ such that for all $\varepsilon > 0$ there is no algorithm that given a graph with $n$ nodes and $m = \Theta(n^\alpha)$ edges detects whether it contains a triangle in $O(m^{4/3-\varepsilon})$ time.
\end{conjecture}

They also considered a weaker form of the conjecture where only some $\Omega(m^{1+\delta})$ lower bound is assumed, and a stronger form with an $m^{\frac{2\omega}{\omega+1}-o(1)}$ lower bound even when $\omega>2$.
However, the above $m^{4/3-o(1)}$ is the more natural and popular hypothesis and it continues to hold even if $\omega=2$.

While the conjecture does not specify the density for which $m^{4/3-o(1)}$ time is required, by a simple high-degree low-degree analysis, one can show that the hardest regime is $m=n^{1.5}$:

\begin{observation}
\label{obs:density}
The Triangle Conjecture is equivalent to the hypothesis that Triangle detection requires $n^{2-o(1)}$ time in graphs with average degree $\Theta(\sqrt{n})$. 
\end{observation}

\begin{proof}
One direction is trivial: If Triangle detection requires $n^{2-o(1)}$ time in graphs with average degree $\Theta(\sqrt{n})$, then for the density regime $m = \Theta(n^{3/2})$ Triangle detection requires $m^{4/3-o(1)}$ time, so the Triangle Conjecture holds. 

For the other direction, suppose that Triangle in graphs with $N$ nodes and $\Theta(N^{1.5})$ edges can be solved in $O(N^{2-\eps})$ time, for some $\eps>0$.
Given a graph on $n$ nodes and $m$ edges as input to Triangle,
let $H$ be the set of nodes of degree $\geq m^{1/3-\delta}$, and let $L=V \setminus H$ be the nodes of degree at most $m^{1/3-\delta}$.

\begin{itemize}
\item To find a triangle that uses any node from $L$, iterate over all $m$ edges $\{u,v\}$ and if one of the endpoints is in $L$, e.g.~$u$, scan its neighborhood and for each $w \in N(u)$ check if $u,v,w$ is a triangle. This takes $O(m \cdot m^{1/3-\delta})$ time.
\item To find a triangle that only uses nodes from $H$ consider the induced graph on these $N=m/m^{1/3-\delta}=m^{2/3+\delta}$ nodes. This graph has only $m=O(N^{1.5})$ edges.
If the number of edges happens to be $o(N^{1.5})$ we can artificially turn it into $\Theta(N^{1.5})$ by simply adding a bipartite graph on $N$ nodes and $N^{1.5}$ edges (this does not introduce any new triangles).
Then, by assumption, we can find a triangle in this graph in time $O(N^{2-\eps})=O((m^{2/3+\delta})^{2-\eps})=O(m^{4/3 +2\delta-2/3 \eps}) = O(m^{4/3-\delta})$ for $\delta<\eps/10$.
\end{itemize}
In both cases we can solve Triangle detection in time $O(m^{4/3-\delta})$, which refutes the Triangle Conjecture. 
\end{proof}

This does not quite prove an equivalence between our hardness assumption and the Triangle Conjecture because we do not know how to reduce the \emph{average} degree $\le \sqrt{n}$ case to the \emph{maximum} degree $\sqrt{n}$ case.

Indeed, this issue arises also in the all-edge version where we do not know how to reduce the general $m=n^{1.5}$ case to the $\sqrt{n}$-degree case.
There, we side-stepped this discussion by starting from other popular conjectures (3SUM and APSP) rather than from a hardness assumption about All-Edges-Triangle itself.
Can we do the same here? Unfortunately, basing the Triangle Conjecture on other popular conjectures such as 3SUM and APSP is a major open question:

\begin{oq}
Can we prove the Triangle Conjecture under other hardness assumptions such as 3SUM or APSP?
\end{oq}

Ever since P{\u{a}}tra{\c{s}}cu's \cite{Pat10} 3SUM-hardness for the all-edge and listing versions of Triangle, it has been a pressing open question to prove the same for detection.
APSP has been connected to Triangle detection in the work of Vassilevska Williams and Williams \cite{williams2018subcubic} but only in a restricted sense: the two problems are subcubic-equivalent for \emph{combinatorial} algorithms in \emph{dense} graphs.
Extending such results to general algorithms or to sparse graphs is a well-known challenge.

As a side result of independent interest, we make progress towards this goal.
We prove the first conditional lower bound for Triangle detection that is based on the hardness of a problem of a very similar flavor to 3SUM and APSP: the Zero-Triangle problem.
Importantly, this hardness continues to hold under the restriction to $\sqrt{n}$-degree graphs, justifying our belief that this is the hard case for triangles.

\begin{definition}[Zero-Triangle]
Given a tripartite graph $G=(A\times B\times C,E)$ with integral edge weights $w:E \to [-W,+W]$ decide if there is a triangle $(a,b,c)\in A \times B \times C$ with total weight $w(a,b)+w(b,c)+w(a,c)=0$.
\end{definition}

Ignoring subpolynomial improvements, there are only two algorithms for this problem.
The first is a brute force over all triples and its running time is $O(|A|\cdot |B| \cdot |C|)$.
In the symmetric setting where $|A|=|B|=|C|=n/3$ this is $O(n^3)$ and it is optimal under both $3$-SUM and APSP conjectures, as long as $W=\Omega(n^3)$ 
\cite{Pat10,WilliamsW13}. 
The second algorithm is faster when $W$ is small enough: It applies the standard exponentiation trick (encoding $w$ as $2^w$) to reduce summation to multiplication and then uses fast matrix multiplication. In the symmetric setting the running time is $O(W \cdot n^\omega)$ and otherwise it is a complicated expression that depends on the rectangular matrix multiplication exponent.
Assuming $\omega=2$, the upper bound simplifies to $(W \cdot N)^{1+o(1)}$ where $N= (|A|\cdot |C| + |A|\cdot |B| + |B|\cdot |C|)$ is an upper bound on the size of the graph.
It is natural to conjecture that these bounds cannot be broken for Zero-Triangle.

\begin{conjecture}[The Strong Zero-Triangle Conjecture]
Zero-Triangle requires $(\min\{WN,|A|\cdot |B| \cdot |C|\})^{1-o(1)}$ time, for any parameters $W,|A|,|B|,|C|=n^{\Theta(1)}$ and where $N=(|A|\cdot |C| + |A|\cdot |B| + |B|\cdot |C|)$. 
\end{conjecture}

While this conjecture is not known to be implied by the $3$-SUM and APSP conjectures (because the existing reductions change the ratio of weight $W$ to number of nodes $n$) its plausibility has the same source.
In fact, it is analogous to the stronger version of the APSP conjecture recently studied by Chan, Vassilevska, and Xu \cite{ChanWX21}.
Notably, Zero-Triangle is a problem that is hard due to the weights and the addition operator and \emph{not} due to the graph structure: the input graph may be assumed to be complete.
Thus, we find it surprising that it explains the hardness of our purely structural subgraph detection problems; in particular it gives a \emph{tight} lower bound for Triangle:

\begin{theorem}
If Triangle detection in graphs with maximum degree $\sqrt{n}$ can be solved in $O(n^{2-\eps})$ time, for some $\eps>0$, then Zero-Triangle with $|A|=n,|B|=|C|=\sqrt{n}$ and $W=\sqrt{n}$ can be solved in $O(n^{2-\eps})$ time, and the Strong Zero Triangle Conjecture is false.
\end{theorem}

\begin{proof}
Given an instance of Zero-Triangle with $|A|=n,|B|=|C|=\sqrt{n}$ and $W=\sqrt{n}$ we construct an unweighted graph as follows.
Each node in $u\in B \cup C$ is copied $6W+1$ times $u_{-3W},\ldots,u_{3W}$ where $u_i$ represents both the node $u$ and the integer value $i$.
A node $a\in A$ has a single copy in the new graph.

An edge of weight $x$ from $a\in A$ to $b \in B$ becomes an edge from $a$ to $b_x$.
An edge of weight $y$ from $a\in A$ to $c \in C$ becomes an edge from $a$ to $c_{-y}$.
On the other hand, an edge of weight $z$ from $b \in B$ to $c \in C$ becomes a matching between the $b_i$ and $c_j$ nodes such that there is an edge between $b_i$ and $c_{i+z}$ for all $i \in [-2W,2W]$.

A zero-triangle $(a,b,c)$ with weights $w(a,b)=x, w(a,c)=y, w(b,c)=z$ becomes a triangle $a,b_x,c_{-y}$. The edges $(a,b_x)$ and $(a,c_{-y})$ exist by definition, and the third edge exists because $-y = x+z$.
By a reverse argument, any triangle in the new graph corresponds to a zero-triangle in the original graph.
\end{proof}

The reduction is rather simple but we find the statement quite interesting.
First, it bases the Triangle Conjecture (and our hardness for $4$-Cycle) on a hardness assumption of a very different nature.
Second, it makes a substantial step towards establishing the Triangle Conjecture  under the more central $3$-SUM or APSP Conjectures.
And third, assuming $\omega=2$, it pinpoints a challenge that one must resolve before making any further progress on Triangle, the lower bound is completely tight for all density regimes (due to Observation~\ref{obs:density}).

\section{Open Questions}

In this paper we have introduced a \emph{short cycle removal} technique and used it to obtain the first conditional lower bounds for $4$-Cycle detection and to demonstrate the optimality of the $O(m^{1+{1/k}})$-time vs. $O(k)$-approximation trade-off for various distance computation problems.
Some of the hardness results are based on the conjectured hardness of the All-Edge Triangle problem (and therefore implied by the 3SUM/APSP Conjectures) and some are based on the hardness of triangle detection.
Let us conclude by highlighting some open questions.

\paragraph{Tight bounds.} Does breaking the longstanding $O(\min(n^2, m^{4/3}))$ upper bound for $4$-Cycle imply a new algorithm for triangle detection? 
Tightening the constants in the exponents of the lower bound of each of the problems we have considered is an interesting open question.
One way to reduce the gaps is by resolving the following conjecture about the relationship between $4$-cycles and dense pieces.

\begin{conjecture}
For all $\eps>0$ there is a $\delta>0$ such that any graph with maximum degree at most $\sqrt{n}$ that has $\geq n^{2+\eps}$ $4$-cycles must have a subgraph on $k$ nodes and $\geq k^{1.5+\delta}$ edges.
\end{conjecture}

A constructive proof of this conjecture (that comes with an efficient algorithm for finding the dense pieces) would establish an $m^{4/3-o(1)}$ lower bound for $4$-Cycle \emph{enumeration} with $n^{o(1)}$ delay.

\paragraph{Which patterns can be detected in linear time?}
In its most basic form, the subgraph isomorphism problem asks if a given graph $G$ on $m$ edges contains a fixed size pattern $H$ as a (not necessarily induced) subgraph. 
It is natural to conjecture that the subgraph isomorphism problem can be solved in $m^{1+o(1)}$ time if and only if $H$ is \emph{acyclic}.
A linear time algorithm for forests follows from the Color-Coding technique \cite{AYZ97}, and this paper proves that all cycles require super-linear time (assuming the hardness of triangles).
Thus, all we have to do is reduce $k$-cycle detection to the detection of any pattern that contains a $k$-cycle.
Such a reduction is known for odd $k$ \cite{DVW21} but not for even. 



\medskip
\paragraph{Acknowledgements}
AA would like to thank Kevin Lewi, Virginia Vassilevska Williams, and Ryan Williams for introducing the fascinating $4$-Cycle problem in his first days at Stanford, and also Arturs Backurs, Greg Bodwin, S{\o}ren Dahlgaard, Mathias B{\ae}k Tejs Knudsen, Aviad Rubinstein, and Morten St{\"{o}}ckel, for stimulating discussions.
We also thank Thatchaphol Saranurak for references on dynamic shortest paths. 

Funding: This work is part of the project TIPEA that has received funding from the European Research Council (ERC) under the European Unions Horizon 2020 research and innovation programme (grant agreement No.\ 850979). This work is supported by the Defense Advanced Research Projects Agency (DARPA) under agreement number HR00112020023. This work is supported by an Alon scholarship and a research grant from the Center for New Scientists at the Weizmann Institute of Science. 

\bibliographystyle{alpha}
\bibliography{main}

\appendix

\section{Reduction from Triangle or 4-Cycle to any $k$-Cycle}
\label{sec:reduction34cycle}

For completeness, we include a proof of the following statement. 
The components of this proof are considered folklore.

\begin{theorem}\label{thm:ckto34}
For any integer $k \ge 3$ one of the following is true:
\begin{itemize}
    \item There is a reduction that given an $m$-edge tripartite graph $G$ runs in $O(m)$ time and constructs a graph $G^\star$ such that the $k$-cycles in $G^\star$ are in 1-to-1 correspondence with the triangles in $G$.
    \item There is a reduction that given an $m$-edge graph $G$ runs in $O(m)$ time and constructs a graph $G^\star$ such that the $k$-cycles in $G^\star$ are in 1-to-1 correspondence with the 4-cycles in $G$.
\end{itemize}
\end{theorem}
Recall that in Triangle detection we can assume without loss of generality that the input graph is tripartite, so this condition makes no big difference. 

By Theorem~\ref{thm:ckto34}, if $k$-Cycle detection can be solved in time $O(m^\alpha)$, for some $\alpha \ge 1$, then either Triangle or 4-Cycle detection can also be solved in time $O(m^\alpha)$. Moreover, if after $O(m^\alpha)$ preprocessing we can enumerate $k$-cycles with $m^{o(1)}$ delay, then the same is true for enumerating either triangles or 4-cycles.

\begin{lemma}\label{lem:cktor}
Let~$r\mid k$ be two positive fixed integers such that~$r$ divides~$k$.
There is a reduction that given an $m$-edge graph $G$ runs in $O(m)$ time and constructs a graph $G^\star$ such that the $k$-cycles in $G^\star$ are in 1-to-1 correspondence with the $r$-cycles in $G$.
\end{lemma}
\begin{proof}
We can assume that~$3\leq r < k$, as otherwise the claim is straightforward.
Let~$G$ be a graph with~$m$ edges, we construct the graph~$G^\star$ by replacing each edge of~$G$ with a path of length~$\frac{k}{r}$ (that is, each edge of~$G$ is subdivided by~$\frac{k}{r}-1$ vertices).
The number of edges in~$G^\star$ is~$\frac{k}{r} m = O(m)$, and constructing~$G^\star$ from~$G$ takes~$O(m)$ time.
We prove the claim by showing that~$G^\star$ contains a $k$-cycle if and only if~$G$ contains an $r$-cycle.

If~$G$ contains an $r$-cycle, then after the subdivision of its edges this $r$-cycle corresponds to a $k$-cycle in~$G^\star$.
On the other hand, any simple cycle in~$G^\star$ can be partitioned into paths of length~$\frac{k}{r}$ corresponding to the full subdivision of edges from~$G$. This holds as the degree of every subdividing vertex is exactly~$2$. Hence, every cycle in~$G^\star$ is of size divisible by~$\frac{k}{r}$ and such cycle of size~$\frac{k}{r}\cdot x$ must correspond to a cycle of size~$x$ in~$G$.
In particular, if~$G^\star$ contains a $k$-cycle then~$G$ contains an $r$-cycle.
\end{proof}

\begin{lemma}\label{lem:ckto3}
Let~$k\geq 3$ be any \textbf{odd} fixed integer.
There is a reduction that given an $m$-edge tripartite graph~$G$ runs in $O(m)$ time and constructs a graph $G^\star$ such that the $k$-cycles in $G^\star$ are in 1-to-1 correspondence with the triangles in $G$.
\end{lemma}
\begin{proof}
Let~$G$ be a tripartite graph with~$m$ edges and vertex sets~$V=A\cup B \cup C$.
We construct~$G^\star$ by replacing every edge of~$G$ with endpoints in~$B$ and~$C$ with a path of length~$k-2$ (that is, we subdivide each edge of~$E(G)\cap \left(B \times C\right)$ by~$k-3$ vertices).
The number of edges in~$G^\star$ and the time to construct it are~$O(m)$.
If~$G$ contains a triangle then~$G^\star$ clearly contains a corresponding $k$-cycle. 
It is left to prove that if~$G^\star$ contains a $k$-cycle then~$G$ contains a triangle.

Denote by~$D^{(i)}$ for~$1\leq i \leq k-3$ the set of all~$i$-th vertices in a subdivision of some subdivided edge.
The graph~$G^\star$ is homomorphic to the $k$-cycle by the partition~$A,B,D^{(1)},\ldots,D^{(k-3)},C$.
Any $k$-cycle in~$G^\star$ must include exactly one vertex in each of the parts $A,B,D^{(1)},\ldots,D^{(k-3)},C$, since the $k$-cycle is not bipartite yet after the removal of any of these parts the remaining graph is homomorphic to a path and hence bipartite.
Due to the degree of each vertex in a part~$D^{(i)}$ being exactly~$2$, such a cycle is necessarily a triangle of~$G$ with one subdivided edge. This follows in a similar manner to the proof of Lemma~\ref{lem:cktor}.
\end{proof}

\begin{proof}[Proof of Theorem~\ref{thm:ckto34}]
Let~$k\geq 3$.
If~$k$ is not a power of~$2$, then it has an odd prime divisor~$p$ and hence we can apply Lemma~\ref{lem:ckto3} to reduce from Triangle detection to $p$-Cycle detection, and then apply Lemma~\ref{lem:cktor} to reduce from $p$-Cycle detection to $k$-Cycle detection, to prove the theorem.
Otherwise, $k\geq 3$ is a power of~$2$ and in particular is divisible by~$4$. Then we can use Lemma~\ref{lem:cktor} to reduce from 4-Cycle detection to $k$-Cycle detection. 
\end{proof}

We note that the components in the proof of Theorem~\ref{thm:ckto34} (and any other previously known technique) do not show that if 4-Cycle detection is linear then so is Triangle detection.
The reason that a similar argument fails is that as a bipartite graph, a 4-cycle can appear between any of the three pairs of parts in~$G$.
On the other hand, we observe that if the original graph~$G$ contains no 4-cycle, then a similar reduction does work.
\end{document}